\documentclass[a4paper,UKenglish,cleveref, autoref,hyperref, thm-state]{lmcs}
\pdfoutput=1
\usepackage[utf8]{inputenc}

\usepackage{lastpage}
\lmcsdoi{21}{3}{8}
\lmcsheading{}{\pageref{LastPage}}{}{}%
{Apr.~09,~2024}{Jul.~23,~2025}{}

\usepackage{orcidlink}
\usepackage{thm-restate}
\usepackage{amssymb}
\usepackage{csquotes}
\def\Nat{\mathbb{N}}
\def\reals{\mathcal{R}^{\geq 0}}

\title{THE CHURCH SYNTHESIS PROBLEM OVER CONTINUOUS TIME}

\author[A.~Rabinovich]{Alexander Rabinovich\lmcsorcid{0000-0002-1460-2358}} 
\author[D.~Fattal]{Daniel Fattal}
\address{Tel-Aviv University}
\email{rabinoa@tauex.tau.ac.il, danielf3141@gmail.com}
\thanks{Supported in part by Len Blavatnik and the Blavatnik Family foundation. D. Fattal was supported in part by  the Israel Science Foundation (grant number 5166651).}

\newcommand{\cA}{\mathcal{A}}

\newcommand{\MSO}{\mathsf{MSO}}
\newcommand{\MLO}{\mathsf{MLO}}

\def\om{\omega}
\def\nat{\mathbb{N}}
\def\O{$\mathcal{O}~$}
\def\Op{$\mathcal{O}$}
\def\I{$\mathcal{I}~$}
\def\Ip{$\mathcal{I}$}
\def\cG{\mathcal{G}}
\def\MLO{\mathit{MSO[<]}}
\def\Sig{\mathit{FVsig}}
\def\reals{\mathcal{R}^{\geq 0}}
\def\XI{X^{\Sigma_{in}}}
\def\Xi2{X^{\Sigma^2_{in}}}
\def\Yi2{Y^{\Sigma^2_{out}}}
\def\Yo {Y^{\Sigma_{out}}}
\def\Xs{X^\Sigma}
\def\Xsig2{X^{\Sigma\times \Sigma}}
\def\Inf{\mathit{Inf}}
\def\sf{\mathit{sf}}
\def\new{\mathit{new}}
\begin{document}

\maketitle

\begin{abstract}
Church’s Problem asks for the construction of a procedure which, given a logical specification
$\varphi(I, O)$ between input $\omega$-strings $I $ and output $\omega$-strings $O$, determines whether there exists an operator
$F$ that implements the specification in the sense that $\varphi(I, F(I))$  holds for all inputs $I$. Büchi and
Landweber provided a procedure to solve Church’s problem for the  specifications in  Monadic Second-Order  logic  and operators
computable by finite-state automata.
We investigate  a generalization of the Church synthesis problem to  the continuous time domain of the non-negative reals.
 We show that
in the  continuous time domain there are phenomena  which are  very different from
the canonical discrete time domain of
the natural numbers.
%seminal  Büchi and Landweber  theorem.
\end{abstract}

\section{Introduction}
\subsection{Church Synthesis Problem}
Let $Spec$ be a specification language and $Pr$ be an implementation language. The following problem is known as the Synthesis problem.

\begin{center}
\fbox{%
  \parbox{0.98\textwidth}{%
  \textbf{Synthesis problem}\\
Input: A specification $S(I,O)\in Spec$.\\
Question: Is there a program $P\in Pr$ which implements it, i.e., $\forall I.\ S(I,P(I))$?
  }%
}
\end{center}

We first discuss the \textbf{classical Church synthesis problem over discrete time $\omega=(\Nat,<)$},  and then explain its generalization to continuous time.

\emph{The specification language} for the Church Synthesis problem is the  Monadic Second-Order Logic of Order ($\MLO$).
Monadic second-order logic is the extension of  first-order logic by monadic predicates and quantification over them.
We use lower case letters $x,y,z,\dots$ to denote first-order variables and  upper case letters $X,Y,Z,\dots$ to denote second-order variables.
The atomic formulas of $\MLO$ are $x<y$ and  $X(y)$. The  formulas are constructed from the atomic formulas by boolean connectives and the first-order and the  second-order quantifiers.
A formula $\varphi(x_0,\dots,x_{k-1},X_0,\dots,X_{l-1})$ is a formula with free variables among $x_0,\dots,x_{k-1},X_0,\dots,X_{l-1}$.

In an $\MLO$ formula $\varphi(X,Y)$, free variables   $X$ and $Y$  range over  monadic predicates over $\Nat$.
Since each such monadic predicate  can be identified with its characteristic $\omega$-string, $\varphi$ defines a binary relation on
  $\omega$-strings. A function $F$ from $\omega$-strings to $\omega$-strings is definable by $\varphi(X,Y)$ iff
  $\forall X \forall Y(Y=F(X) \leftrightarrow\varphi(X,Y))$.
%\rightsquigarrow

The \emph{implementation} languages defines functions from  $\omega$-strings  to $\omega$-strings. Such functions are called operators. A machine that computes an operator
at every moment $t\in \mathbb{N}$ reads an input symbol $X(t)\in \{0, 1\}$, updates its internal state, and produces an output
symbol $Y (t) \in \{0, 1\}$. Hence, the output $Y(t)$ produced at $t$ depends only on inputs
symbols $X(0), X(1),\dots, X(t)$. Such operators are called causal operators (C-operators);   if
the output $Y(t)$ produced at $t$ depends only on inputs symbols $X(0), \dots, X(t-1)$,
the corresponding operator is called a strongly causal (SC-operator).

Another property of interest is that the machine computing $F$ is  finite-state.
In  light of Büchi’s proof \cite{Bu62}  of the expressive equivalence of $\MLO$  and finite automata,
$F$  is finite-state if and only if it is $\MLO$-definable.

The following problem is known as the Church Synthesis problem over the discrete time of  naturals $\omega:=(\Nat,<)$.
\begin{center}
\fbox{%
  \parbox{0.98\textwidth}{%
  \textbf{Church Synthesis problem}\\
Input: an $\MLO$ formula $\Psi(X,Y)$.\\
Question: Is there  a C-operator $F$ such that
$ \forall X\Psi(X,F(X))$ holds over $\omega:=(\Nat,<)$?
  }%
}
\end{center}
%\bigskip
B\"{u}chi and Landweber \cite{BL} proved that the Church synthesis problem is
computable. Their main theorem can be stated as: % follows:
\begin{thm}
\label{BLthm}
Let $\Psi(X,Y)$ be an $\MLO$ formula.
\begin{enumerate}
    \item Determinacy: exactly one of the following holds for $\Psi$:
    \begin{enumerate}
        \item There is a C-operator $F$ such that  $\omega\models \forall X.\ \Psi(X,F(X))$.
        \item There is a SC-operator $G$ such that $\omega\models \forall Y.\ \neg \Psi(G(Y),Y)$.
    \end{enumerate}
    \item Decidability: it is decidable whether 1 (a) or 1 (b) holds.
    \item Definability:
    \begin{enumerate}
        \item If 1 (a) holds, then there is an $\MLO$ formula $U$ that defines a C-operator which implements $\Psi$.
        %\item Similarly for 1 (b).
        \item If 1 (b) holds, then there is an $\MLO$ formula $U$ that defines a SC-operator $G$ such that $\omega\models \forall Y.\ \neg \Psi(G(Y),Y)$.
    \end{enumerate}
    \item Computability: There is an algorithm such that for each $\MLO$ formula $\Psi(X,Y)$:
    \begin{enumerate}
        \item If 1 (a) holds, constructs an $\MLO$ formula $\varphi(X,Y)$ that defines $F$.
       % \item Similarly for 1 (b).
       \item If 1 (b) holds, constructs an $\MLO$ formula $\varphi(X,Y)$ that defines $G$.
    \end{enumerate}
\end{enumerate}
\end{thm}
Recall that a formula $U(X,Y)$ defines an operator  $F$ if  $\forall X\forall Y (Y=F(X) \leftrightarrow U(X,Y))$.
In the above theorem  we use ``C-operator definable by an $\MLO$ formula'' instead of an equivalent ``operator computable by a finite state transducer (automaton).'' This  allows us to lift this theorem to continuous time where finite state transducers were not defined.

The proof of Theorem \ref{BLthm} is based on the fundamental connections between Logic, Games and Automata.
\par
Classical  Automata theory results show an equivalence between three fundamental
formalisms: logical nets, automata and $\MLO$
logic.
Trakhtenbrot \cite{Tr95} suggested a program  to extend the classical results from discrete time 
  to other time domains.
%Pardo, Rabinovich and
%Trakhtenbrot  \cite{DAT04}  obtained some initial results in this
%programme.

In this work we investigate the Church synthesis problem over continuous time.

$\MLO$ has a natural interpretation over the non-negative reals $(\reals,<)$.

Shelah \cite{She75} proved that $\MLO$  is undecidable over the reals if we allow quantification over
arbitrary predicates. In Computer Science, however, it is natural to restrict predicates  to finitely variable (non-Zeno)
predicates.

A signal  is a function $f$ from $\reals$  into a finite set $\Sigma$  such that  $f$ satisfies the finite variability (non-Zeno) property: there exists an unbounded increasing sequence
$\tau_0=0<\tau_1<\cdots$  such that $f$ is constant on every open interval $(\tau_i,\tau_{i+1})$.
In other words, $f$
is finitely variable if it   has finitely
many discontinuities in any bounded sub-interval of $\reals$.

 We denote the structure for $\MLO$ over $(\reals,<)$ with the finite variability  predicates  as the interpretation for the monadic variables,
by $\Sig$; $\MLO$ over $\Sig$ is decidable. The  $C$-operators   and $SC$-operators from  signals to signals are defined naturally: $F$ is a $C$ (respectively, SC) operator
if the output $F(X)$ at $t\in \reals$  depends only on $X$ in the interval $[0,t]$ (respectively, $[0,t)$). We investigate the following synthesis problem:
\begin{center}
\fbox{%
  \parbox{0.98\textwidth}{%
  \textbf{Church Synthesis problem for continuous time}\\
Input: an $\MLO$ formula $\Psi(X,Y)$.\\
Question: Is there   a C-operator $F$ such that
$ \forall X.\ \Psi(X,F(X))$ holds in $\Sig$?
  }%
}
\end{center}
\subsection{Results}
It turns out that  the discrete case and the case for continuous time are very different.
Here are our main results.
\begin{restatable}[Indeterminacy]{thm}{indeter}
\label{indeterminancy}  The synthesis problem for continuous time is indeterminate.\\
There exists  an $\MLO$ formula $\Psi(X,Y)$ such that:
\begin{enumerate}
    \item There is no C-operator $F$ such that  $\Sig\models \forall X.\ \Psi(X,F(X))$.
    \item There is no SC-operator $G$ such that  $\Sig\models \forall Y.\ \neg\Psi(G(Y),Y)$.
\end{enumerate}
\end{restatable}
\begin{restatable}[Dichotomy Fails]{thm}{dichfail} 
\label{dich}  %The synthesis problem for continuous time is indeterminate.\\
There exists an $\MLO$ formula  $\Psi(X,Y)$ such that:
\begin{enumerate}
    \item There is a C-operator $F$ such that  $\Sig\models \forall X.\ \Psi(X,F(X))$.
    \item There is a SC-operator $G$ such that  $\Sig\models \forall Y.\ \neg\Psi(G(Y),Y)$.
\end{enumerate}
\end{restatable}
\begin{restatable}[Undefinability]{thm}{undef} \label{Undefinability}
% \begin{thm}[Undefinability]
% \label{Undefinability}
 There exists an $\MLO$ formula $\Psi(X,Y)$ such that:
\begin{enumerate}
    \item There is no $\MLO$-definable C-operator $F$ such that  $\Sig\models \forall X.\ \Psi(X,F(X))$.
    \item There is a C-operator $F$ such that  $\Sig\models \forall X.\ \Psi(X,F(X))$.
\end{enumerate}
\end{restatable}
%\begin{exa}Consider a specification  $\Psi(X,Y)$ that states:
%\begin{enumerate}
%  \item $Y$  is zero (false) everywhere except  on finitely many points, where $Y$ has value one (true), and
%  \item There is a nonsingular interval $[t_1,t_2]$ where $X$ is constant and $Y$ has value one  at an internal point $t_3\in (t_1,t_2)$.
%  \end{enumerate}
%This specification is easily formalized in $\MLO$. Note that if $F$ is an $\MLO$ definable operator  and $X$ is constant in $[t_1,t_2]$,
%then $Y=F(X)$ is constant in $(t_1,t_2)$. Therefore, there is  no $\MLO$ definable operator that implements $\Psi$.
%However,  the following (strongly) causal operator $F$ implements $\Psi$: $F(X)$ is one at the  points
%$1, 1+\frac{1}{2}, \cdots, 1+\sum_{i=1}^{k} \frac{1}{2^i}$ and zero everywhere else, if $X$ is not continuous  at  $1, 1+\frac{1}{2}, \cdots, 1+\sum_{i=1}^{k-1}    \frac{1}{2^i}$ and $X$ is continuous  at $1 +\sum_{i=1}^{k}  \frac{1}{2^i}$, where ``$X$ is  continuous at $t$'' stands for  ``there are  $t_1<t<t_2$ such that $X$
%is constant in $[t_1,t_2]$.'' Indeed, since $X$ is finitely variable, it is impossible that $X$  is not continuous at all points  $1+\sum_{i=1}^{m} \frac{1}{2^i}$  for $m\in \Nat$. Hence, there is  $k$ such that $X$ is not   continuous at  $1, 1+\frac{1}{2}, \cdots, 1+\sum_{i=1}^{k-1}    \frac{1}{2^i}$ and  $X$  is  continuous at $1 +\sum_{i=1}^{k}  \frac{1}{2^i}$. For this $X$, $F(X)$ will be one at $1 +\sum_{i=1}^{m}  \frac{1}{2^i}$ for $ m=1,\cdots, k$, and zero at all other points. Is is easy to verify that $F$ is strongly causal.
%\end{exa}
Theorem \ref{Undefinability} leads us to two natural questions: (1) Is it decidable whether a formula is implementable by an $\MLO$-definable C-operator, and if so, is a formula defining an  implementation computable? (2) Is it decidable whether a formula is implementable by some C-operator?\\
The following theorems answer these questions:

\begin{restatable}[Computability of Definable Synthesis]{thm}{dsynth} \label{th:def-c}
Given an $\MLO$ formula $\Psi(X,Y)$, it is decidable whether there exists an $\MLO$-definable C-operator $F$ such that  $\Sig\models \forall X.\ \Psi(X,F(X))$ and if so, there is an algorithm that constructs an $\MLO$ formula that defines $F$.
\end{restatable}

%\begin{restatable}[Decidability of Synthesis]{thm}{defsynt}
%\label{mainthm}
%Every even integer greater than 2 can be expressed as the sum of two primes.
%\end{restatable}

\begin{restatable}[Decidability of Synthesis]{thm}{synth}
\label{mainthm}
Given  an $\MLO$ formula $\Psi(X,Y)$, it is decidable whether there  exists a C-operator $F$ such that $\Sig\models \forall X.\ \Psi(X,F(X))$.
\end{restatable}
%\synth*
The proofs of   Theorems \ref{indeterminancy}-\ref{th:def-c} are not difficult. The last theorem is our main result.  %All the proofs can be found in \cite{Fat23}.
%As mentioned above, we will use games, but of more complex nature than the parity games used in the proof of Theorem 1.1. Our games are also turn based, and involves considerations of time.
\subsection{Related Works}

As far as we know, the Church synthesis problem over continuous time was considered only in \cite{JenkinsORW11}.
This paper considered only the \emph{right continuous} signals.
A signal $f$ is right continuous if there is
an unbounded increasing sequence
$\tau_0=0<\tau_1<\cdots$  such that $f$ is constant on every interval $[\tau_i,\tau_{i+1})$. Note that there is no  $\MLO$ -definable 
 non-constant SC-operator
from the right continuous signals to the right continuous signals.  Computability of Definable Synthesis (like  Theorem \ref{th:def-c}) was proved there. In \cite{JenkinsORW11}  we also considered $\mathit{MSO[<,+1]}$ - the extension of $\MLO$ by a  metrical plus  one function.
This is a very strong logic and it is undecidable over the reals; yet, it is decidable over any  bounded subinterval $[0,n]$ of the reals for every $n\in \Nat$. We proved  Computability of Definable Synthesis for $\mathit{MSO[<,+1]}$  for every bounded interval.

Timed automata \cite{AD} are popular  models of real-time systems.   Timed automata
accept timed words  - $\omega$-sequences annotated with real numbers. These timed words   are  different from   signals.
Two-player games on timed automata have been considered (see e.g., \cite{Hen,FTM}).
These games are concurrent games. Both players suggest a move simultaneously. Due to concurrency, timed automata games are indeterminate.
Patricia Bouyer  \cite{Bou05}  writes that there is a ``prolific literature on timed games.''   Unfortunately, ``one paper = one definition of timed games which are parameterized by $\cdots$."
% by: 
%``symmetrical games or not, take care of Zeno (and Zeno-like) behaviours or not, turn-based games or not,
%etc.''

%Similar to our games, timed automata games  take into account time divergence.
We  were unable to deduce our results from the
results about     timed automata games.
\subsection{Structure of the Paper}
The paper is organized as follows.
Section  \ref{sect:prel} contains preliminaries.
Section \ref{sect-simple-th} proves Theorems \ref{indeterminancy} - \ref{Undefinability}.
In Section \ref{sec:def-synth} computability of definable synthesis problem is proved.
Decidability of the Church synthesis problem is proved in Sections \ref{sect:ch-rc}-\ref{sect:ch-fv}.
Section  \ref{sect:ch-rc} provides a proof for  right continuous signals and contains
the main ideas. Section \ref{sect:ch-fv} extends the proof to FV-signals. This proof is longer and more technical,
but does not require new insights.
Section \ref{sect:conc} contains the Conclusion. The appendix contains a proof of a lemma.
\section{Preliminaries}\label{sect:prel}
\subsection{Monadic Second-Order Logic of Order}
\subsubsection{Syntax}
The language of Monadic Second-Order Logic  of Order ($\MLO$)
has  individual variables,  monadic second-order
variables,
a binary predicate $<$ , the usual propositional connectives and
first and second-order quantifiers $\exists^1$ and $\exists^2$.
 We use
$t,~v $ for individual variables and $X,~ Y$
  for second-order variables.
It will be clear from the context whether  a quantifier is the first
or the second-order, and we will drop the superscript.
We use  standard abbreviations, in particular, ``$\exists !$" means
\enquote{there is a unique.} 

The  atomic  formulas of $\MLO$ are formulas of the form:
$t<v$ and $X(t)$. Formulas are constructed from atomic formulas by using
logical connectives and first and second-order
quantifiers.

We write $\psi(X,Y,t,v)$ to indicate that the free variables of a formula $\psi$ are  among  $X, ~Y,~t,~v$.
\subsubsection{Semantics} 
A structure $K = \langle A,B, <_K \rangle $ for $\MLO$
 consists of a set $A$ partially ordered by $<_K$ and a
set $B$ of monadic functions from $A$ into $\{0,1\}$. The letters $\tau,t,s,x,y$ will range over the
elements of $A$ and capital bold letter $\textbf{X,Y,Z},\dots$  will range over the elements of $B$. We will not distinguish between a subset of $A$ and
its characteristic function. The satisfiability relation $K, \tau_1, \dots,
\tau_m,  \textbf{X}_1 \dots  \textbf{X}_n \models \psi(t_1, \dots ,t_m,~X_1, \dots ,X_n)$
  is defined in a standard way. We sometimes use $K\models \psi(\tau_1, \dots
\tau_m,~ \textbf{X}_1 ,\dots,  \textbf{X}_n)$ for $K, \tau_1, \dots,
\tau_m,  \textbf{X}_1, \dots,  \textbf{X}_n  \models \psi(t_1, \dots, t_m,~X_1, \dots, X_n)$.

We will be interested in the following structures: 
\begin{enumerate}
  \item  The structure $\omega=\langle \Nat,2^\Nat,<_\Nat \rangle$ where $2^\Nat$ is the set of all monadic functions from $\Nat$ into $\{0,1\}$. 
  \item The signal structure $\Sig$ is defined as $\Sig = \langle \reals,SIG,<_R\rangle$, where $SIG$ is the set
of finitely variable  boolean  signals.
  \item A signal is right continuous if  
    there exists an unbounded increasing sequence
$\tau_0=0<\tau_1<\cdots$  such that $f$ is constant on every  half-open  interval $[\tau_i,\tau_{i+1})$.
   The right continuous signal structure $Rsig$ is defined as $Rsig = \langle \reals, RSIG, <_R \rangle$
where $RSIG$ is the set of right continuous boolean  signals.
\end{enumerate}

\noindent \textbf{{Terminology:}} For a FV   (or RC)  signal $f$ we say that $f$ is \emph{constant} or \emph{continuous} at $\tau$
if there are $\tau_1<\tau<\tau_2$ such that $f$ is constant in the interval $[\tau_1,\tau_2]$. If $f$ is not constant at $\tau$, we say that it \emph{jumps} at $\tau$ or it is \emph{discontinuous} at $\tau$. In particular, every signal jumps at $\tau=0$. 
%1. The structure $\omega=\langle \Nat,2^\Nat,<_\Nat \rangle$ where $2^\Nat$ is the set of all monadic functions from $\Nat$ into $\{0,1\}$. 
%2. The signal structure $\Sig$ is defined as $\Sig = \langle \reals,SIG,<_R\rangle$, where $SIG$ is the set
%of finitely variable boolean signals.\\
%3. The right continuous signal structure $Rsig$ is defined as $Rsig = \langle \reals, RSIG, <_R \rangle$
%where $RSIG$ is the set of right continuous boolean signals.
%\par 
\subsection{Coding}
Let $\Delta$ be a finite set. We can code a function  from a set $D$ to $\Delta$ by a tuple of unary predicates on $D$.
This type of coding is standard,  and we shall use explicit  variables which range over such mappings and expressions of the form ``$F(u)=d$''
(for $d\in \Delta$) in MSO-formulas, rather than their codings.

Formally, for each finite set $\Delta$ we have second-order variables $X_1^\Delta, X_2^\Delta, \dots$
which range over the functions  from $D$ to  $\Delta$, and atomic formulas
$X_i^\Delta(u)=d $  for $d\in \Delta$ and a first-order variable $u$ \cite{trakhtenbrotfinite}.
Often, the type of the second-order variables will be clear from the context (or unimportant) and we drop the superscript $\Delta$.
\subsection{Definability} Let
$\phi(X)$ be an $\MLO$  formula.
We say that  a language (set of predicates)   %$L\subseteq \mathbb{R}^{\geq 0}$
$L$ is
definable by $\phi(X)$ in a structure $K$  if $L$ is the   set of monadic predicates which  satisfies $\phi(X)$ in $K$.

\medskip
\noindent {\it Example  (Interpretations of Formulas).}
\begin{enumerate}
\item Formula  $\forall t_1\forall t_2.~t_1<t_2\wedge (\neg \exists
t_3.~t_1<t_3<t_2) \to (X(t_1)\leftrightarrow \neg X(t_2))$
 defines the $\omega$-language $\{(01)^{\omega}, ~(10)^{\omega}\}$ in the
structure $\omega$ and it  defines the set of all signals in the signal structures, since every set $X\subseteq \mathbb{R}^{\geq 0}$ satisfies this formula.
  \item   
Formula $\exists Y.~ \exists t'. Y(t') \wedge (\forall t.~X(t)\to Y(t))
\wedge (\forall t_1
\forall t_2. t_1<t_2 \wedge Y(t_1)\wedge Y(t_2) \to \exists t_3. t_1<t_3<t_2 \wedge \neg
 Y(t_3))$ defines in the structure
$\omega$    the set of strings in which  between any two occurrences of 1 there is an occurrence of 0.  In the FV-signal structure,
the above formula defines the set of signals that receive value 1 only at isolated points. The formula defines the empty
language under the right continuous signal interpretation.
\end{enumerate}
  In the above examples, all formulas have one free
second-order variable and they define languages over the  alphabet $\{0,1\}$.
 A formula $\psi(X_1, \dots ,X_n)$ with
  $n$ free second order variables  defines a language over the  alphabet $\{0,~1\}^n$.
We say that a $\omega$-language
is definable if it is definable by a monadic formula in the structure $\omega$.

Let $\psi(\XI, \Yo )$ be an $\MLO$ formula. We say that $\psi$  defines an operator $F$ if for all $X$ and $Y$: $\psi(X,Y)$
if and only if $Y=F(X)$.

 \begin{defi}[Speed independent signal language] Let $L$ be a signal language.  $L$ is speed independent if for each order preserving bijection $\rho : \reals \to \reals$:  \begin{center} $X\in L$ if and only if $\rho(X)\in L$.
 \end{center}
 \end{defi}
 It is clear:
\begin{lem}
 The $\MLO$ definable languages  are speed independent.    
\end{lem} 
 Moreover, 
 \begin{lem}[The output signal jumps only if the input signal jumps] \label{InJumpsWhereOutJumps}
 If $F$ is an $\MLO$ definable operator  and $F(\textbf{X})$ is discontinuous at $t>0$, then $\textbf{X}$ is discontinuous at $t$.
 \end{lem}
 \begin{proof}
 Assume that  $F$ is  definable by $\psi(X,Y)$ and $\textbf{Y}=F(\textbf{X})$ is discontinuous at $t>0$.
 If $\textbf{X}$ is continuous at $t$, then there is an interval $(a,b)$ such that $t$ is the only point of discontinuity of $(\textbf{X,Y})$ in $(a,b)$.
 Let $\rho : \reals \to \reals$  be an order preserving bijection that is the identity map outside of $(a,b)$ and maps  $t$ to $t'\neq t$. 
 Now $(\rho(\textbf{X}),\rho(\textbf{Y}))$ satisfies $\psi$.  Note that $\rho( \textbf{{X}}) =\textbf{X}$. Hence, by the functionality of $F$,  $\rho (\textbf{Y})=\textbf{Y}$, and this is a contradiction, because 
 the only point of discontinuity of $\textbf{Y}$ in $(a,b)$ is $t$, and the only point of discontinuity of $\rho(\textbf{Y})$ in $(a,b)$ is $t'\neq t$.
\end{proof}

\subsection{Parity Automata}
A deterministic parity automaton is a tuple $\cA = (Q,\Sigma ,\delta, q_{\mathit{init}}, pr)$ that consists of the following components:
\begin{itemize}
  \item $Q$ is a finite set. The elements of $Q$ are called the states of $\cA$.
  \item $\Sigma$  is a finite set called the alphabet of $\cA$.
  \item $\delta: ~Q\times \Sigma\rightarrow Q$ is   the transition function of $\cA$.
  \item $q_{\mathit{init}}\in Q$ is  the initial state.
  \item $pr:~Q\rightarrow \Nat$ is the priority function.
\end{itemize}
 An input for $\cA$ is an $\om$-sequence $\alpha =\sigma_0\sigma_1\cdots $ over $\Sigma$.
 The run of $\cA$ over $\alpha$   is an infinite sequence $\rho = r_0,r_1,r_2,\cdots $  of states, defined as follows:
 \begin{itemize}
   \item $r_0:=q_{\mathit{init}}$.
   \item $r_{i+1}:=\delta(r_i,\sigma_i)$ for $i\geq 0$.
   
 \end{itemize}
\emph{Acceptance condition:} Let  $\Inf(\rho)$ be the set of states that occur infinitely often in $\rho$.
$\cA$ accepts exactly  those runs $\rho$ for which $\min( pr(q)\mid q\in \Inf(\rho))$  is even.  
An $\om$-string $\alpha$  is accepted if the run over $\alpha$ is accepted.
An $\om$-language of $\cA$ is the set of $\om$-strings accepted by $\cA$.

The following fundamental result is due to Büchi:

\begin{thm} An $\om$-language is definable by an $\MLO$ formula  if and only if it is accepted by a deterministic parity
automaton.
\end{thm}

\section{Proofs of Theorems \ref{indeterminancy} - \ref{Undefinability}}\label{sect-simple-th}
\subsection{Indeterminacy}\label{section:Indeterminacy}
\indeter*
%\label{section:Indeterminacy}
To prove that the synthesis problem is indeterminate, for each $x\in \reals$ consider the following signal
 $$\delta_x(t):=
  \begin{cases}
    1, & \text{if } t = x  \\
    0  & \text{else }
  \end{cases}$$
Recall that a signal $X$ jumps at $t$ if it is not continuous at $t$. Let\\ %Say that an interval is left open if it contains no minimum.
%\par
$$\Psi(X,Y):=\text{$\exists t>0$ such that $X$ is constant on $(0,t]$ and $Y$ jumps at $t$}$$
\begin{proof}[Proof of Theorem \ref{indeterminancy}] \hfill
\begin{enumerate}
    \item Assume that there  exists a C-operator $F$ such that $\Sig\models \forall X.\ \Psi(X,F(X))$. 
     
     Notice that for each $x>0$, since $\delta_x$ jumps at $x$ and is constant on $(0,x)$, we have the following equivalence
    $$\Psi(\delta_x,F(\delta_x)) \mbox{ iff }  \text{$\exists t\in(0,x)$ such that $F(\delta_x)$ jumps at  $t$}.$$
    Therefore, $\Sig\models \Psi(\delta_1,F(\delta_1))$ implies that
    $$\textbf{(*)}\ \  \text{ $\exists t\in(0,1)$ such that $F(\delta_1)$ jumps at  $t$}.$$
    By non-Zenoness, there exists a \textbf{minimal} $t\in (0,1)$ that satisfies $\textbf{(*)}$. Moreover, 
    \begin{enumerate}
        \item By causality, since $\delta_{1}=\delta_{t}$ on $[0,t)$ then $F(\delta_{1})=F(\delta_{t})$ on $[0,t)$.
        \item By minimality, $F(\delta_1)$ is constant on $(0,t)$. So by (a), both $F(\delta_{1}),F(\delta_{t})$ are constant on $(0,t)$.
        \item By $\Psi(\delta_t,F(\delta_t))$, $F(\delta_t)$ jumps at  some $x\in (0,t)$. That   contradicts  (b).
    \end{enumerate}
    \item Assume that there  exists a SC-operator $G$ such that $\Sig\models \forall Y.\ \neg\Psi(G(Y),Y)$. By applying logical equivalences we obtain:\begin{center}
                                                   $\neg \Psi(X,Y)\mbox{ iff } \forall t>0.~ \text{$X$ jumps on $(0,t]$  {or} $Y$ is continuous at $t$}$.
                                                \end{center}
    %$\neg \Psi(X,Y)\mbox{ iff } \forall t>0.~ \text{$X$ jumps on $(0,t]$  {or} $Y$ is continuous at $t$}$.
    Therefore, for each $x\in \reals$:\\
    $\neg \Psi(G(\delta_{x}),\delta_{x})\mbox{ iff } \forall t>0.~G(\delta_{x})\text{ jumps on $(0,t]$  {or} $\delta_{x}$ is continuous at $t$}$,
    but since $\delta_{x}$ is  {not} continuous at $x$, we obtain  
    $G(\delta_{1})\text{ jumps on $(0,1]$}$.  
    
    %So $\delta_1$ is constant on the maximal left open interval that $G(\delta_1)$ is constant on. Therefore the first jump point of $G(\delta_1)$ is in $(0,1]$.
    Let $a$ be the first positive jump point of $G(\delta_{1})$. Since $\delta_1=\delta_{a/2}$ on $[0,a/2)$,   by strong causality  of $G$, we obtain that $G(\delta_1)=G(\delta_{a/2})$ on $[0,a/2]$. Thus, $G(\delta_{a/2})$ is constant on $(0,a/2]$ and $\delta_{a/2}$ jumps at $a/2$. %in $(0,a/2]$.
     We obtained  a contradiction to $\neg \Psi(G(\delta_{a/2}),\delta_{a/2})$. \qedhere
\end{enumerate}
\end{proof}

\subsection{Failure of dichotomy}
\label{section:NonDichotomy}
%In this section, we prove
\dichfail*
% In this section, we prove
% that some $\MLO$ formulas are \enquote{not dichotomous}, i.e. those formulas can be implemented by a C-operator and their negations can be implemented by an SC-operator. Moreover, those operators are $\MLO$ definable.

For the proof of Theorem \ref{dich}, consider the formula $\Psi$, the C-operator $F$, and the SC-operator $G$ given below:
$$\Psi(X,Y):=\forall x(X(x)\leftrightarrow Y(x))$$
$$F(X):=X$$
$G(X)(0):=0$, and for $t>0$:
$$G(X)(t):=\begin{cases}
    1-a  & \text{ if $X$ has a constant value $a$ over $(0,t)$}\\
    1 & \text{ else}
\end{cases}$$
It is clear that $F$ is a $C$-operator that implements $\Psi$. Below we prove (2).
%\begin{enumerate}[label=\alph*]
\begin{description}  %[label=\alph*]
    \item[a] $G$ is a strongly-causal operator: suppose $\textbf{X}_1=\textbf{X}_2$ on $(0,t)$. Either both $\textbf{X}_1,\textbf{X}_2$ have a jump point in $(0,t)$ or both are constant on  $(0,t)$.\\
    In the first case, by Non-Zenoness, there exists $t_0\in (0,t)$ such that $t_0$ is the \textbf{first} jump point of both $\textbf{X}_1,\textbf{X}_2$ in $(0,t)$. 
    Therefore, both $\textbf{X}_1,\textbf{X}_2$ have a constant value $a$ over $(0,t_0)$. Thus, by definition $$\forall t\geq  0.~G(\textbf{X}_1)(t)=G(\textbf{X}_2)(t)=\begin{cases}
       1-a &  t\in (0,t_0)\\
       1 & else
    \end{cases}$$
    The proof for the second case is similar.
    
    \item[b] $G$ implements $\neg\Psi$: for each signal $\textbf{Y}$ there exists a non-empty interval such that $\textbf{Y}\neq G(\textbf{Y})$ on it. Thus, $\textbf{Y}\neq G(\textbf{Y})$.
%\end{enumerate}
\end{description}
\subsection{Undefinability}
\label{section:Undefinability}
%Here we prove
\undef*
% We prove that some $\MLO$ formulas are implemented by a C-operator but \textbf{not} by a definable one. For the proof of theorem \ref{Undefinability}, 
Consider the formula
$\Psi(X,Y):{=} \text{ $Y$ jumps inside  }(0,\infty)$.\\
It is implemented by
$F(X)(t):=\begin{cases}
  1, & \text{if } t=1\\
    0,  & \text{else }
\end{cases}$\\
When $X$ is constant on $[0,\infty)$, $Y$ should jump at $t>0$, where $X$ is continuous.
However, by Lemma  \ref{InJumpsWhereOutJumps}, no $\MLO$-definable operator implements $\Psi$.
Hence, $\Psi$ is implemented by a C-operator, but  it is not implemented by a definable operator.

\section{Computability of the Definable Synthesis} \label{sec:def-synth}
In \cite{JenkinsORW11} the computability of the definable synthesis problem over the  right continuous signals was proved.

\begin{thm}[Computability of the Definable Synthesis over the RC Signals]\label{th:def-c-rc}
Given an $\MLO$ formula $\Psi(X,Y)$, it is decidable whether there  exists an $\MLO$-definable C-operator $F$ such that  $RC\models \forall X.\ \Psi(X,F(X))$ and if so, there is an algorithm that constructs an $\MLO$ formula that defines $F$.
\end{thm}
In this section we extend this theorem to  finite variability signals.

\dsynth*
Our proof is a variation of the proof of Theorem \ref{th:def-c-rc}. 
In the next subsection we recall a  relationship between finite variability signals and $\omega$-strings.
In subsection \ref{sub:def:operators}, we discuss definable operators.
Finally,  in subsection \ref{subs:alg}, we present an algorithm for the  definable synthesis problem.
\subsection{FV-signals and \texorpdfstring{$\om$}{ω}-strings}
We recall the relationship between finite variability signals and $\omega$-strings.

Let $X$ be a finite variability signal over $\Sigma$ and let $\widetilde{\tau}:=0=\tau_0<\tau_1<\cdots <\tau_n<\cdots$ be an $\omega$-sequence of reals such that $X$ is constant on every interval $ (\tau_i,\tau_{i+1})$ and $\lim \tau_i=\infty$.  We say that such $\widetilde{\tau}$ is a sample sequence   for $X$. In other words, an unbounded
$\om$-sequence is a sampling sequence for $X$ if it contains all
the points of discontinuity of $X$.

We define:
\begin{itemize}
\item $ D(X,\widetilde{\tau})$ be an $\omega$-string $(a_0,b_0) (a_1,b_1) \dots (a_i,b_i)\dots $ over $\Sigma\times \Sigma$ defined as $a_i:=X(\tau_i) $ and $b_i$ is the value of $X$ in $(\tau_i,\tau_{i+1})$.
    We say that $D(X,\widetilde{\tau})$ \textbf{represents} $X$.
\item $D(X):=\{D(X, \widetilde{\tau})\mid \widetilde{\tau}$ is a sample sequence  for $X\}$.   
\item For a signal language $L$ define $D(L):=\bigcup_{X\in L} D(X)$.
\end{itemize} 
Observe that if an $\om$-string $U$ represents a signal $X$ and $\rho$ is an order preserving bijection over non-negative reals, then $U $ represents $\rho(X)
$.
\begin{itemize}
\item 
For an $\omega$-string $U=(a_0,b_0) (a_1,b_1) \dots (a_i,b_i)\dots $  over $\Sigma\times \Sigma$ and an unbounded $\om$-sequence $\widetilde{\tau}:=0=\tau_0<\tau_1<\cdots <\tau_n<\cdots$  of reals define a finite variability signal:

 $ FV(U,\widetilde{\tau})$ to be $a_i$ at $\tau_i$ and $b_i$ in the interval $(\tau_i,\tau_{i+1})$. 
 \item $FV(U):=\{FV(U,\widetilde{\tau})\mid \widetilde{\tau}$ is  an unbounded $\om$-sequence $\}$.
 \item For an $\om$-language $L$ define:  $FV(L):=\bigcup_{U\in L} FV(U)$.
\end{itemize}  
The following theorem  was proved in \cite{AR}:
  \begin{thm}\label{thm:LangToSig}
    A finitely variable signal language $S$ is $\MLO$-definable if and only if there
is an $\MLO$-definable $\omega$-language $L$ such that $ S = FV(L)$.
Moreover, (1) there is an algorithm that from an $\MLO$ formula $\varphi(\Xs)$ computes an $\MLO$ formula $\varphi^D(\Xsig2)$ such that if $S$ is a signal language definable by $\varphi$ then $D(L)$ is definable by $\varphi^D$, and 
(2)  there is an algorithm that from an $\MLO$ formula $\varphi(\Xsig2)$ computes an $\MLO$ formula $\varphi^{FV}(\Xs)$ such that if $L$ is an $\om$-language definable by $\varphi$ then $FV(L)$ is definable by $\varphi^{FV}$.
\end{thm} 
Note also that if $L$ is a definable signal language, then $L=FV(D(L))$.  However,  for a definable $\om$-language $L$, it might be the case that 
$L\subsetneq D(FV(L))$.
%$L\neq D(FV(L))$.
\begin{defi}[Stuttering Equivalence]
    Two $\om$-strings are stuttering equivalent if there is a signal  $X$ which is represented by both of them. Equivalently, the sets of signals represented by these strings coincide.
\end{defi}
\begin{defi}[Stuttering Free]
An $\om$-string $U$  is stuttering free  if there is a signal $X$ such that (1) or (2) holds.
\begin{enumerate}
\item $X$ has  infinitely many points of discontinuity and $U=D(X, \widetilde{\tau})  $ where $\widetilde{\tau}$  is a   sample sequence that contains exactly all points of discontinuity  of $X$, or
\item $X$ has $k$ points of discontinuity $0=\tau_0 <\tau_1 <\cdots <\tau_{k-1}$, and $U=D(X, \widetilde{\tau})  $, where  $\widetilde{\tau}:= 0=\tau_0 <\tau_1 <\cdots <\tau_{k-1} <\tau_k< \cdots$, where $\tau_i$ (for $i<k$) are the points of discontinuity of $X$,
and $\tau_j$ for $j\geq k$ are arbitrary.
\end{enumerate}
    
\end{defi}
\begin{lem}\hfill
\begin{enumerate}
    \item For every signal $X$, all strings in $D(X) $ are stuttering equivalent. 
\item Every $\om$-string $V$ is stuttering equivalent to a unique stuttering free string $U$.
\end{enumerate}

\end{lem}
%Two $\om$-strings are stuttering equivalent if there is $X$ which is represented by both of them. Equivalently, a the sets of signals represented by these string coincide.

%A string $s $ is stuttering free if there is $X$ same signal
Stuttering free $\om$-strings can be characterized as follows:
\begin{lem} 
   An $\om$-string $U$ over $\Sigma  \times\Sigma$ is stuttering-free  if it satisfies the following condition:  if  $i$-th letter of $U$  is $(a,b)$ and  $i+1$-th letter is $(b,b)$, then  for all $j>i$,   $j$-th letter of $U$  is $(b,b)$. 
\end{lem}
%
% An $\om$-string $U$ over $\SIgma\times \Sigma$ is stuttering-free  if it satisfies the following condition:  if  $i$-th letter of $U$  is $(a,b)$ and  $i+1$-th letter is $(b,b)$, then  for all $j>i$,   $j$-th letter of $U$  is $(b,b)$.
Hence,
\begin{lem}[Stuttering free strings are definable]
There is an $\MLO$ formula that defines the set of stuttering free $\om$-strings over $\Sigma \times \Sigma$.
\end{lem}

\subsection{Definable Operators} \label{sub:def:operators} 
% \begin{lem}[FV definable operators]
% Let $\Psi(\XI  ,\Yo)$ be an  $\MLO$ formula. If $\Psi$ defines a finite variability operator  $F$ then:  $F(X)$ is not continuous  at $t$ only if $X$ is not continuous at $t$.
% \end{lem}
Let $\alpha(\Xi2 , \Yi2 )$ be an $\MLO$ formula. It defines an  $\om$-language over  $\Sigma^2_{in}\times \Sigma^2_{out}$.
To an $\om$-string over  $\Sigma^2_{in}\times \Sigma^2_{out}$ in a natural way corresponds an $\om$-string over $(\Sigma_{in}\times \Sigma_{out})^2$ and an $\om$-string over $\Sigma_{in}\times \Sigma_{out}$. Hence, we can consider $\alpha$ as defining an $\om$-language over $\Sigma_{in}\times \Sigma_{out}$. We say that $\alpha$ defines a C-operator (respectively, SC-operator), if its language is the graph
of this operator.

 % If $F$ is a definable  operator on  finite variability signals, then:  $F(X)$ is not continuous  at $t$ only if $X$ is not continuous at $t$.

\begin{lem} \label{deinable-sf}
    Let $\alpha(\Xi2 , \Yi2 )$ be a formula such that:
    \begin{enumerate}
        \item If $(\Xi2,\Yi2) $ satisfies  $\alpha$, then $\Xi2$ is stuttering free.
        \item  For every stuttering free $X$ 
        there is a unique $Y$ such that $(X,Y)$ satisfies $\alpha$.
    \end{enumerate}
    Then, $\alpha^{FV}$ defines an operator on    
    signals.  Moreover,  if in addition $\alpha$ defines C-operator
    (respectively, SC-operator) on $\om$-strings, then   $\alpha^{FV}$ defines C-operator
    (respectively, SC-operator) on FV-signals.
    \end{lem}
%\begin{proof}
%    
%\end{proof}
\begin{lem} 
 \label{prop:from-fvtod}
If there is a definable  C-operator that implements
$\Psi(\XI,\Yo )$ over FV-signals, then there is a definable C-operator that implements
 $\Psi^D$ (over $\om$).
\end{lem}
\begin{proof}
Assume that $\alpha$ defines a  C-operator that implements
$\Psi(\XI,\Yo )$. We claim that  
 $\alpha^D$ defines a C-operator that implements
 $\Psi^D$ (over $\om$).
 First, note that $\alpha^D  $ implies $ \psi^D $.
 Now, for every $U\in (\Sigma^2_{in})^\om$ there is a FV signal 
 $X$ and $\widetilde{\tau}$ such that
 $U=D(X,\widetilde{\tau})$. Let $Y$ be such that $ \alpha(X,Y)$ holds.  Since $Y$ changes only when $X$ changes, $\widetilde{\tau}$  is also a  sample sequence  for $Y$.  Let $V:=D(Y,   \widetilde{\tau})$. It is clear that $\alpha^D(U,V)$. Therefore, for every $U$ there is $V$ such that $\alpha^D(U,V)$. It remains to show that such $V$ is unique.
 Indeed, if  there are $V_1$ and $V_2$ such that
 both $(U,V_1)$ and $(U,V_2)$ satisfy $\alpha^D$,
 then  for arbitrary $\widetilde{\tau}$,
 $FV (  (U,V_1)  , \widetilde{\tau})$ and 
  $FV (  (U,V_2)  , \widetilde{\tau})$ satisfy
  $\alpha$. However,   $FV (   V_1  , \widetilde{\tau})\neq FV (   V_2  , \widetilde{\tau}) $
   and both are $\alpha$ related to $FV(U, \widetilde{\tau})$. This contradicts the fact that $\alpha$ defines an operator
   on signals.

  Causality and strong causality are immediate.
\end{proof}
Recall that  Lemma \ref{InJumpsWhereOutJumps} states: if $F$ is a definable  operator on  finite variability signals, then  $F(X)$ is not continuous  at $t$ only if $X$ is not continuous at $t$.
\begin{lem}
\label{lem:C-operator-prop}
Let $\Psi(\XI  ,\Yo)$ be in $\MLO$ and let $F$ be an $\MLO$ definable C-Operator (or SC-operator).\\
Denote $\Psi^*(\XI  ,\Yo ):= $
$$\Psi \land \text{$\forall t\in (0,\infty)$ if ${X}$ is continuous at $t$, then so is $ {Y}$}\land $$
$$\text{if ${X}$ is left-continuous at $t$, then so is ${Y}$.}$$
Then $$\text{ $F$ implements $\Psi \iff F$ implements $\Psi^{*}.$}$$
\end{lem}

% \begin{lem}
% \label{lem:SC-operator-prop}
% Let $\Psi(\widehat{X},\widehat{Y})$ be in $\MLO$ and let $G$ be an $\MLO$ definable SC-Operator.\\
% Denote $\Psi^{\dagger}(\widehat{X},\widehat{Y}):=$
% $$\Psi\land \text{$\forall t\in (0,\infty)$, if $\widehat{X}$ is continuous at $t$ then so is $\widehat{Y}$} \land $$
% $$\text{$\widehat{Y}$ is left continuous on $(0,\infty)$}$$
% then $$\text{ $G$ implements $\Psi \iff G$ implements $\Psi^{\dagger}$}$$
% \end{lem}
\subsection{Algorithm}\label{subs:alg}
Now, we are ready to  present an algorithm that verifies whether
there is a definable C-operator that implements  an $\MLO$  formula
over FV-signals, and if so, it computes an $\MLO$ formula that defines such an operator.

\begin{center}
\fbox{%
  \parbox{0.98\textwidth}{%
\begin{center}  \textbf{Algorithm}
\end{center}
Let $\Psi(\XI, \Yo) $ be an $\MLO$ formula  and  let $\Psi^*$ be defined as in Lemma \ref{lem:C-operator-prop}. 

Use Büchi-Landweber algorithm to decide whether  
there is a C-operator that implements $(\Psi^*)^D$ over $\om$.

If the answer is ``NO'' then there is no $\MLO$ definable C-operator that implements $\Psi$.

If $\alpha$  defines a C-operator that implements $(\Psi^*)^D$ over $\om$, then let
$$\alpha_{\sf}:=\alpha(\Xi2,\Yi2 ) \land \mbox{``$X$ is stuttering free.''}$$ 

Let  $\beta:=(\alpha_{\sf})^{FV}$. Then,  $\beta$  defines a C-operator that implements $\Psi$.
}%
}
\end{center}

\vspace{11pt}

%Correctness of the algorithm.
\begin{proof}[Correctness of the algorithm.]  ~ \hfill

  If the answer is ``NO'' the correctness follows  from 
Lemma \ref{prop:from-fvtod} and Lemma \ref{lem:C-operator-prop}.

If the algorithm returns $\beta$, the correctness follows from Lemma \ref{deinable-sf} and Lemma  \ref{lem:C-operator-prop}.
  
\end{proof}
% If the answer is ``NO,'' the correctness follows  from 
% \Cref{prop:from-fvtod} and \Cref{lem:C-operator-prop}.
%
% If the algorithm returns $\beta$, the correctness follows from \Cref{deinable-sf} and \Cref{lem:C-operator-prop}.
%
The algorithm that decides  whether there is a  SC-definable operator on FV-signals that implements $\Psi$ is almost the same:
replace everywhere ``C-operator'' by ``SC-operator.''

\begin{rem}[Non-causal uniformization]

A  uniformizer  for a binary relation $R\subseteq A\times B$ is a maximal subrelation $R^*$ which is the graph of a partial function.
  Equivalently, let us call the ``cross-section at $a$" for $a\in A$, the set $R_a=\{b\mid (a,b)\in R\}$, then a uniformizer for $R$ is a relation $R^*\subseteq R$  such that for all~$a\in A$, if $R_a$ is non-empty, then there exists exactly one~$b$ such that $(a,b)\in R^*$. 
  
%A class of relations on a structure  $\cM$  has the uniformization property if ...

For every $\MSO[<]$ formula $\Psi(X,Y)$, there  is an  $\MSO[<]$ formula  $\Phi(X,Y)$ such that the (binary) relation definable by $\Phi(X,Y)$ on $\om$ is a uniformizer for the (binary) relation definable by $\Psi(X,Y)$ on $\om$ 
 \cite{LS1998,Rab07}. However, over   signals,  the relation ``there is $t>0$ such that $X$ is continuous at $t$, but  $Y$ is discontinuous at $t$''
 is $\MSO[<]$-definable, yet it 
  has no $\MSO[<]$-definable uniformizer. 
  
% Arguments similar to those in the proof of Theorem \ref{th:def-c}  show that it is decidable whether a relation definable by an $\MSO[<]$ formula $\Psi(X,Y)$ on signals,
%  has an $\MSO[<]$-definable uniformizer.
%  
  Arguments analogous to those presented in the proof of Theorem \ref{th:def-c} demonstrate that it is decidable whether a relation defined by an $\MSO[<]$ formula $\Psi(X,Y)$ on signals has an $\MSO[<]$-definable uniformizer.
  
  Let $\Psi(X,Y)$ be an $\MSO[<]$ formula  and let $$\Psi^\dag(X,Y):=\Psi(X,Y)\land \text{$\forall t\in (0,\infty)$ if ${X}$ is continuous at $t$, then so is $ {Y}$}.$$
  Then $\Psi $ has an $\MSO[<]$-definable uniformizer iff $\forall X(\exists Y \Psi \rightarrow \exists Y \Psi^\dag)$ holds in the signal structure.
  For every $\Psi $:  $\Psi^\dag$ has an $\MSO[<]$-definable uniformizer which is computable from~$\Psi$.
 % Moreover,   if $\Phi$ is  an $\MSO[<]$-definable uniformizer for $\Psi$ then  $\Phi$ is a uniformizer for $\Psi^\dag$, and (2) $\Psi^\dag$,
 % $\Phi$ is computable from $\Psi$.
%
\end{rem}

\section{Games and the  Church Problem}\label{sect:ch-rc}
%Our aim is to show that the Church Synthesis problem is computable for FV signals.
%Here, we provide a  proof of decidability of Church's synthesis problem over right continuous signals.
%It contains all the main ideas but it is  slightly simpler than the proof for FV signals.
%In the next section we outline its extension to FV signals.
Our objective is to demonstrate that the Church Synthesis problem is computable for FV signals. In this section, we present a proof of the decidability of Church's synthesis problem for right-continuous signals. This proof encompasses all the main ideas,  but is slightly simpler than the proof for FV signals. In the following section, we outline its extension to FV signals.

Recall that a timed $\om$-word (or timed $\om$-sequence) over an alphabet   $\Sigma$ is an $\om$-sequence $(a_0,\tau_0)(a_1,\tau_1)\cdots (a_i,\tau_i)\cdots $, 
where $a_i\in \Sigma$ and $0=\tau_0<\tau_1<\cdots <\tau_i<\cdots$ is an unbounded $\om$-sequence of reals. 
Such a sequence represents  a right continuous signal $X$ defined as $X(\tau):=a_i$ for $\tau\in[\tau_i,\tau_{i+1})$. Note that every right continuous
signal is represented by a timed sequence.

We start with a reduction  of the Church Synthesis problem to  a game.
Let $\Psi(X,Y)$ be a formula.
Consider the following two-player game $\cG_\Psi$.  During the play,  \I constructs an    input signal (equivalently, timed word) $X$ and \O constructs  an output $Y$.
%\marginpar{Give the name to the game}

Round 0. Set $t_0:=0$.
\begin{enumerate}\item

 \I chooses a value    $a_0\in \Sigma_{in}$  which holds at $t_0$ and for a while after  $t_0$.
 %\item  \O responses  with $b_0\in \Sigma_{out}$.
 %\item
  %        \I chooses $a^d_0$ to hold for a while after $t_0$;
           \item \O responds with a timed $\omega$-word  (equivalently, with a signal) $Y'_0$ on $[t_0,\infty)$.

      Now, \I has two options:

      A) \I accepts. In this case, the play ends. \I has constructed $X$ which has $a_0$ at all $t\geq 0$, and \O  has constructed $Y:=Y'_0$   -  the signal chosen by \Op.

      B)  \I interrupts at some $t_1>t_0$ by choosing
    %  \begin{enumerate}
     % \item
     $a_1\neq a_0$ which holds at $t_1$ and for a while after  $t_1$.  % and going to step 1 of round 1.
      %\item $a^d_1\neq  a^d_0$ and going to step 3 of round 1.
      %\end{enumerate}
      Signal   $X_0$ constructed by \I is $a_0$   in  $[t_0,t_1)$  and is $a_1$ at $t_1$,
       %while in (b) $a^d_0$ at $t_1$ and $a^d_1$ for a while after $t_1$;
        $Y_0$ constructed by \O  is  the restriction of $Y'_0$
        on the interval $[t_0,t_1)$. 
    %Signal   $X'_0$ constructed by \I is $a_0$   in  $[t_0,t_1)$  and is $a_1$ at $t_1$,
%       %while in (b) $a^d_0$ at $t_1$ and $a^d_1$ for a while after $t_1$;
%        $Y'_0$ constructed by \O  is  the restriction of $Y'_0$
%        on the interval $[t_0,t_1)$. 
        %
        Now,  proceed  to step 1 of round 1.

\end{enumerate}
Round $i+1$.
\begin{enumerate}
%\item
%\O chooses $b_{i+1}\in \Sigma_{out}$ to hold at $t_{i+1}$.
%\item
%\I
% chooses $a^d_{i+1}$ to hold for a while after $t_{i+1}$;
\item
  \O chooses  a timed $\omega$-word (equivalently,  a signal) $Y'_{i+1}$ on $[t_{i+1},\infty)$.
\item
  Now, \I has two options:

     % A) \I accepts. In this case, the play ends. The signal $X$ constructed by \I coincides with  $X'_i$ on $[0,t_{i+1})$ and equals  $a_{i+1}$ on $[t_{i+1},\infty)$. The signal $Y$ constructed by \O coincides with $Y'_i$ on $[0,t_{i+1})$ and with $Y'_{i+1}$ on $[t_{i+1},\infty)$.
      A) \I accepts. In this case, the play ends. The signal $X$ constructed by \I coincides with  $X_i$ on $[0,t_{i+1})$ and equals  $a_{i+1}$ on $[t_{i+1},\infty)$. The signal $Y$ constructed by \O coincides with $Y_i$ on $[0,t_{i+1})$ and with $Y'_{i+1}$ on $[t_{i+1},\infty)$.

       B) \I interrupts at some $t_{i+2}>t_{i+1}$ by choosing
     % \begin{enumerate}
       $a_{i+2}\neq a_{i+1}$   for  $X$   at $t_{i+2}$    and for a while after it.
      %  \In this case round $i+2$ starts. % (in this case $X$ is left discontinuous at $t_{i+1}$) or by
     % \item    $a_{t+1}=a^d_i$  and    $a^d_{i+1}\neq  a^d_i$. In this case round $i+2$ starts at step 3 by a choose of \O  (in this case $X$ is right discontinuous at $t_{i+1}$).
     % \end{enumerate}
   %  In this case,  $X'_{i+1}$ is  the  signal that has value $a_{i+2}$ at $t_{i+2}$,  $a_{i+1} $ on $[t_{i+1},t_{i+2})$  and extends the signal $X'_i$ constructed in  the previous steps on $[0,t_{i+1})$,
%                                           and  $Y'_{i+1}$  coincides with $Y'_i$ on  $[0,t_{i+1})$ and  coincides with $Y'_{i+1}$  on $[t_{i+1}, t_{i+2})$.
%                                                                                                                                       Start  round $i+2$. %
  In this case,  $X_{i+1}$ is  the  signal that has the  value $a_{i+2}$ at $t_{i+2}$,  $a_{i+1} $ on $[t_{i+1},t_{i+2})$  and extends the signal $X_i$ constructed in  the previous steps on $[0,t_{i+1})$,
     and  $Y_{i+1}$  coincides with $Y_i$ on  $[0,t_{i+1})$ and  coincides with $Y'_{i+1}$  on $[t_{i+1}, t_{i+2})$.
    Start  round $i+2$. %

\end{enumerate}
    % Let $X'_{i+1}$ be the  signal that has value $a_{i+1} $ on $[t_{i+1},t_{i+2})$  and extend the signal $X'_i$ constructed in  the previous steps on $[0,t_{i+1})$,
%     and let $Y$ be the signal    constructed by  \O.
After $\omega$ rounds, if $\lim t_i=\infty$, we define $X$ (respectively, $Y$) to be the signal that  coincides  with $X_i$ (respectively, with $Y_i)$ on $[t_i,t_{i+1})$ for every $i$.

\emph{Winning conditions.}     \O wins the play iff $\Psi(X,Y)$ holds or \I interrupts infinitely often and
the duration of the play is finite, i.e., $\lim t_n< \infty$, where $t_n$ is the time of $n$-th interruption.

       Note that in this game, player \I has much more power than player \Op. Player \O provides a signal untill infinity and \I observes it and can interrupt it at any  time  moment she chooses.
       Nevertheless, we  have the following lemma.

  \begin{lem}[Reduction to games] \label{lem:red}
    \O wins $\cG_\Psi$  if and only if there is a C-operator that implements $\Psi$.
  \end{lem}
   \begin{exa} Assume that $\Psi$ expresses that $Y$ has a discontinuity at some  $t>0$ where $X$ is continuous.
  Let us describe a winning strategy for \Op. At the  $i$-th move \O chooses a signal which has value $ X(t_i) $ at $[t_i,t_i+\frac{1}{2^i})$ and
   $\neg X(t_i)$ at $t\geq t_i+\frac{1}{2^i}$.
  It is clear that this is a winning strategy for \Op. Indeed, in order to ensure that $\Psi$ fails,
 at the $i$-th round, \I should interrupt at $t_{i+1}$ before $t_{i}+\frac{1}{2^{i} }$.  Therefore, the duration of such a play is less than $\sum 2^{-i}=2<\infty$. Hence, \O wins.
  \end{exa}

 % \begin{exam} Assume that $\Psi$ expresses that $Y$ has a discontinuity at some  $t>0$ where $X$ is continuous.
%  Let us describe a winning strategy for \Op. At $i$-th move it chooses a signal which has 1 at the point  $t_i+1/2^i$ and $0$ at the other points.
%  It is clear that this is a winning strategy for \Op. Indeed, in order to ensure that $\Psi$ fails,
% at the $i$-th move, \I should interrupt at $t_i$ before $t_{i-1}+1/2^{i-1}$. Therefore, the duration of such play $<\sum 2^{-i}=1<\infty$.
%  \end{exam}
 \begin{proof}[Proof of Lemma \ref{lem:red}]\hfill
 
  $\Rightarrow$-direction. Let $\sigma$ be an \Op-winning strategy.
 Consider the play where \I ignores \Op's moves and constructs an input signal $X$.  Let $Y$ be the signal   constructed by \O in  this play, This defines a $C$-operator, and
 $\Psi(X,Y) $ holds.

$\Leftarrow$-direction.  Assume that $F$ is a $C$-operator that implements $\Psi$. Define a winning strategy for \O in an obvious way.

 Round 0.
   If  \I chooses $a_0$, then \O replies by $F(X_0)$ the   output produced to the input signal which is   $a_0$ everywhere. If
 \I interrupts at $t_1>t_0$ and chooses $a_1\neq a_0$, then \O replies by $Y'_1$, where  $Y'_1$ is the restriction of  $Y_1=F(X_1) $ on $[t_1, \infty)$, where
  $X_1$ is $a_0$ on $[0,t_1)$ and $a_1$ on $[t_1, \infty)$. %is the output of $F$ on the signal which is $a_0$ at $t:=0$
 %and $a^d_0$ at every $t>0$.

 Continue in a similar way for the other rounds.  This strategy is winning for \Op. \end{proof}
\subsection{An automaton Game}
Let $\cA$ be a deterministic parity  automaton equivalent to $\Psi^D$. We  modify the  game $\cG_\psi$  as follows. Player \O chooses signals (equivalently, timed words) over the set  $Q:=Q_\cA$ of  states of $\cA$, instead of timed words over $\Sigma_{out}$. Of course, the
timed words  should be consistent with the  behavior of $\cA$.
%\textbf{Explain consistent}
A play is winning for \Op,  if the constructed  $\om$-word satisfies the  acceptance conditions of  $\cA$.
Alternatively, we can pretend that \O chooses  timed words over $\Sigma_{out}$. However, we are interested in the corresponding timed
words  over $Q$.

Here is a modified game $\cG_\cA$ on the automaton $\cA$ equivalent to $\Psi^D$.

\fbox{
\begin{minipage}{0.9\linewidth}

\begin{center}
 $\cG_\cA$ Game
\end{center}
\begin{description}
\item[Round 0]
 Set $t_0:=0$ and $q_0:=q_{\mathit{init}}$.
\begin{enumerate}\item

 \I chooses a value    $a_0\in \Sigma_{in}$  which holds at $t_0$ and for a while after  $t_0$.
 %\item  \O responses  with $b_0\in \Sigma_{out}$.
 %\item
  %        \I chooses $a^d_0$ to hold for a while after $t_0$;
           \item               %        \O responds with a timed $\omega$-word  (equivalently, with a signal $Y_0$ on $[t_0,\infty)$).

            \O responds with a  timed $\omega$-word (equivalently, with a signal) $Y'_0=(q_1,\tau_1), \cdots  (q_k,\tau_k)\cdots $ over $Q$ such that $\tau_1>t_0$ and
$q_{j+1} =\delta_\cA(q_j,(a_0,b_j))$ for some $b_j\in \Sigma_{out}$ and all $j\geq 0$.   % where
%  $q_0:=q_{\mathit{init}}$.

      Now, \I has two options:

      A) \I accepts. In this case, the play ends. %\I has constructed $X$ which has $a_0$ at all $t\geq 0$, and 
      \O  has constructed $Y:=Y'_0$. %   -  the signal chosen by \Op.

      B)  \I interrupts at some $t_1>t_0$ by choosing
    %  \begin{enumerate}
     % \item
     $a_1\neq a_0$ which holds at $t_1$ and for a while after  $t_1$.  % and going to step 1 of round 1.
      %\item $a^d_1\neq  a^d_0$ and going to step 3 of round 1.
      %\end{enumerate}
   % Signal   $X'_0$ constructed by \I is $a_0$   in  $[t_0,t_1)$  and is $a_1$ at $t_1$,
       %while in (b) $a^d_0$ at $t_1$ and $a^d_1$ for a while after $t_1$;
     %   $Y'_0$ constructed by \O  is  the restriction of $Y_0$
      %  on the interval $[t_0,t_1)$.
        %
        % Set $q_1:=Y_0(t_1)$
        %  and
             Proceed  to step 1 of round 1.
\end{enumerate}

%
%Replace in  step 2 ``\O responds \dots'' by  ``\O responds with a  timed $\omega$-sequence $Y_0=(q_1,\tau_1), \cdots  (q_k,\tau_k)\cdots $ over $Q$ such that $\tau_1>t_0$ and
%$q_{i+1} =\delta_\cA(q_i,(a_0,b_i))$ for some $b_i\in \Sigma_{out}$ and all $i\geq 0$,'' where
%  $q_0:=q_{\mathit{init}}$.

\item[Round $i+1$]   Set $q_0:=Y'_i(t_{i+1})$.
 \begin{enumerate}
%\item  
%\O chooses $b_{i+1}\in \Sigma_{out}$ to hold at $t_{i+1}$.
%\item
%\I
% chooses $a^d_{i+1}$ to hold for a while after $t_{i+1}$;
\item
  %\O chooses  a timed $\omega$-word (equivalently,  a signal) $Y'_{i+1}$ on $[t_{i+1},\infty)$.
  \O chooses  a timed $\omega$-sequence     $Y'_{i+1}:=(q_1,\tau_1), \cdots  (q_k,\tau_k) \cdots $ over $Q$ such that $\tau_1>t_{i+1} $
  and  $q_{k+1} =\delta_\cA(q_k,(a_{i+1},b_k))$ for some $b_k\in \Sigma_{out}$ and all $k\geq 0$. % where $q_0$ is the last state of $Y'_i$.''
\item
  Now, \I has two options.

      A) \I accepts. In this case, the play ends.
       %The signal $X$ constructed by \I coincides with  $X'_i$ on $[0,t_{i+1})$ and equals  $a_{i+1}$ on $[t_{i+1},\infty)$. 
       The signal $Y$ constructed by \O coincides with $Y'_j$ on $[t_j,t_{j+1})$ for $j\leq i$ and with $Y'_{i+1}$ on $[t_{i+1},\infty)$.

       B) \I interrupts at some $t_{i+2}>t_{i+1}$ by choosing
     % \begin{enumerate}
       $a_{i+2}\neq a_{i+1}$   for  $X$   at $t_{i+2}$    and for a while after it.
      %  \In this case round $i+2$ starts. % (in this case $X$ is left discontinuous at $t_{i+1}$) or by
     % \item    $a_{t+1}=a^d_i$  and    $a^d_{i+1}\neq  a^d_i$. In this case round $i+2$ starts at step 3 by a choose of \O  (in this case $X$ is right discontinuous at $t_{i+1}$).
     % \end{enumerate}
     In this case, 
      %$X'_{i+1}$ is  the  signal that has value $a_{i+2}$ at $t_{i+2}$,  $a_{i+1} $ on $[t_{i+1},t_{i+2})$  and extends the signal $X'_i$ constructed in  the previous steps on $[0,t_{i+1})$,
    % and  $Y'_{i+1}$  coincides with $Y'_i$ on  $[0,t_{i+1})$ and  coincides with $Y'_{i+1}$  on $[t_{i+1}, t_{i+2})$.
    start  round $i+2$. %

\end{enumerate}

%If play lasted $\om$ rounds set $Y:=Y_i$ for $t\in [t_i,t_{i+1})$.

%
% In 
%  Step 1, replace ``\O chooses \dots''  by  ``\O chooses  a timed $\omega$-sequence     $Y'_{i+1}:=(q_1,\tau_1), \cdots  (q_k,\tau_k) \cdots $ over $Q$ such that $\tau_1>t_{i+1} $
%  and  $q_{k+1} =\delta_\cA(q_k,(a_{i+1},b_k))$ for some $b_k\in \Sigma_{out}$ and all $k\geq 0$, where $q_0$ is the last state of $Y'_i$.''
 \item[Winning Condition] 
 If the play lasted $\om$ rounds, set $Y:=Y'_i$ for $t\in [t_i,t_{i+1})$.  For plays of  finite length  $Y$ was defined by option  A). 
 \O wins
  if either the $\om$-sequence of states of $Y$ produced by \O is accepted by $\cA$, or there were $\om$ rounds, but $\lim t_i$ is finite.
 \end{description}
\end{minipage}
}

\bigskip
The following is immediate.

\begin{lem}[$\cG_\Psi$ and $\cG_\cA$ are equivalent]\label{lem:eq-aut-game}
  \O wins in $\cG_\Psi$ if and only if \O wins in $\cG_\cA$.
\end{lem}

%\marginpar{Automata - replace $\Sigma_{out}$ moves by $Q$ moves. Modified game should be moved to here and its equivalence to the original game stated here}
 Now, we are going to restrict \O moves.
The first restriction is a time schedule.
When  \O chooses a timed sequence   $(q_0,\tau_0) (q_1,\tau_1)\cdots (q_i,\tau_i)\cdots$    we require that $\tau_{i+1}-\tau_i=\tau_{j+1}-\tau_j$  for all $i,j$. % _the distance between consecutive jumps in this signal should  be the same.
We call such a sequence (signal) \emph{simple timed}.
Another restriction is on the sequence of values.
We require that this sequence is ultimately periodic.

In the next subsection we describe the  ultimately periodic sequences that we need.

\subsection{Ultimately Periodic Sequences}\label{sub:ult-periodic}
\begin{defi}[Ultimately periodic strings]\label{def:up}
    An $\omega$-string $w$  is ultimately periodic if $w=uv^\om$ for some (finite) strings $u$ and $v$; $v$ (respectively, $u$)  is the periodic (respectively, pre-periodic or lag)  part of $w$; the length of $v$ (respectively, of  $u$) is period (respectively, lag) of $w$.
% A finite string is periodic with a lag $d$ and a period $p$ if it is a prefix of an ultimately periodic $\om$-string with
%lag $d$ and period $p$.
\end{defi}
We denote by $u[n]$ the $n$-th letter of a string $u$ and by $ u[k,m)$ the set of letters $\{ u[i]\mid k\leq i<m)\}$.
\begin{defi}
[$\equiv_{state}$-equivalence on finite strings]%\label{def:eq}  % Unnused multiply defined label
Let $\cA$ be an automaton over a set  $Q$ of states  and let $u,v$ be  finite  strings over $Q$.
% Let $\Sigma$ be an alphabet.
We say that $u$ is equivalent to $v$  (notation $u\equiv_{state}v$) if and only if
\begin{enumerate}
\item % $u$ and $v$ have the same first state and the same last state. %first state of $u$Inf(u)=Inf(v), i.e.,   a state  appear infinitely often in $u$ iff it appears infinitely often in $v$.
    The first state of $u$ is the same as the first state of $v$, and the  last state of $u$ is the same as the last state of $v$.
\item For every $m$ there is $n$ such that $u[m]=v[n]$  and the states that appear in $ u[0,m)$ are the same as the states that appear in $v[0,n)$.   %and symmetrically
\item Symmetrically, for every $n$ there is $m$ such that $u[m]=v[n]$  and the states that appear in $ u[0,m)$ are the same as the states that appear in $v[0,n)$.
    \item    $u$  is run of $\cA$  if and only  if   $v$ is a run of $\cA$.
    \end{enumerate}
%\marginpar{Maybe simplify?}
%
%Finite strings over $Q$ are equivalent if they contains the same letters,  their first letters are the same
%and their last letters are the same.
 \end{defi}
For the sake of formality,  we have to use notations $u\equiv_{state}^{\cA}v$
 %$u\equiv_{state}^{Q,\Sigma}v$
  to emphasize that this equivalence relation is parameterized by  $\cA$.  % (finite) sets $Q$ and $\Sigma$.
   %finite sets $Q$ and $\Sigma$.
    However, $\cA$  %$Q$ and $\Sigma$
   will be clear from the context or unimportant and we will use
 $\equiv_{state}$  instead of  $\equiv_{state}^{\cA}$.
% $\equiv_{state}^{Q}$.
 
 The next lemma summarizes properties of $\equiv_{state}$-equivalence:  % and will be proved in the Appendix.
 \begin{lem}[Properties of $\equiv_{state}$]\label{lem:prop-equiv} For every $\cA$: % and $\Sigma$:
 \begin{enumerate}
   \item $\equiv_{state}$ has finitely many equivalence classes.
   \item $\equiv_{state}$ is a congruence with respect to concatenation, i.e., if  $u_1\equiv_{state}v_1$ and  $u_2\equiv_{state}v_2$, then
   $u_1u_2\equiv_{state}v_1v_2$.
   \item For every $\om$-string $\pi$   there are $u_0,\dots, u_n,\dots$ such that
   $\pi=u_0\cdots u_n\cdots$ and $u_i\equiv_{state} u_j \equiv_{state} u_iu_{i+1}$ for all $i,j>0$ and
   $u_0\equiv_{state} u_0u_1$.

   \end{enumerate}
 \end{lem}\begin{proof}
            (1) and (2) are immediate. (3)  will be proved in the Appendix.
          \end{proof}

\begin{defi}[Equivalence on $\om$-strings]
%\label{def:eq} % Unnused multiply defined label
Let $\cA$ be an automaton over a set $Q$ of states,  and let $u,v$ be  $\omega$-strings over $Q$. %Let $\Sigma$ be an alphabet.
We say that $u$ is equivalent to $v$ iff
\begin{enumerate}
\item $\Inf(u)=\Inf(v)$, i.e.,   a state  appears infinitely often in $u$ iff it appears infinitely often in $v$.
\item for every $m$ there is $n$ such that $u[m]=v[n]$  and the states that appear in $ u[0,m)$ are the same as the states that appear in $v[0,n)$,   %and symmetrically
\item Symmetrically, for every $n$ there is $m$ such that $u[m]=v[n]$  and the states that appear in $ u[0,m)$ are the same as the states that appear in $v[0,n)$.
     \item $u$ is a run of  $\cA$,  %over  $\Sigma$
    if and only  if $v$ is a run of $\cA$.
\end{enumerate}
%\marginpar{Maybe simplify?}
%%
%Finite strings over $Q$ are equivalent if they contains the same letters,  their first letters are the same
%and their last letters are the same.
\end{defi}
The next lemma is immediate. It deals with  the  concatenations of $\om$-sequences of finite strings.
\begin{lem}\label{lem:-m-concat}
  if $u_i\equiv_{state}v_i$ for $i\in \om$, then $ u_0u_1\dots $ is equivalent to $ v_0v_1\dots $
\end{lem}
%\marginpar{redefine UP}
From  Lemma \ref{lem:-m-concat} and Lemma \ref{lem:prop-equiv} we obtain:
\begin{lem}\label{lem:up}
  For every automaton $\cA$ over a set $Q$ of states  there is a finite set of ultimately periodic strings $UP(Q)$ (=$UP_\cA(Q)$ computable from $\cA$) 
   %such that for every automaton $\cA$ over the alphabet $\Sigma$ and the set of states  $Q$
  such that for  every $\om$-word  $\pi$ over  $Q$ there is $\rho=uv^\omega\in UP(Q)$ which is  equivalent to $\pi$,  and
  $v\equiv_{state} vv$  and
   $u\equiv_{state} uv$. Moreover, there are $u_0, u_1 ,\cdots ,u_i,\cdots$ such that
   $\pi=u_0 u_1 \cdots u_i\cdots$  and $u_0\equiv_{state} u$ and $u_i\equiv_{state} v$ for every  $i>0$.
  \end{lem}
 \begin{proof}
   For each $\equiv_{state}$-equivalence class (over finite strings)  choose a (shortest) string $u$ in it. This string  is called a representative of its 
    $\equiv_{state}$-equivalence class. Let $u_1,\dots ,u_k$ be the representatives.
   A $\equiv_{state}$-equivalence class $C$ is idempotent  if   $uv\equiv_{state}u$ for every $u,v\in C$, equivalently $uu\equiv_{state}u$ for
   an element $u\in C$. Define $UP(Q):=\{u_iu_j^\om\mid u_i$ and $u_j$ are representatives, and $u_j$ is  in  an idempotent class
   and $u_iu_j   \equiv_{state}u_i\}$. Lemma \ref{lem:prop-equiv}(3) implies that every $\om$-string is %$\equiv_{state}$-equivalent 
     equivalent to a string
   in $UP(Q)$. Lemma \ref{lem:prop-equiv}(3) also implies the ``moreover part'' of Lemma \ref{lem:up}.
 \end{proof}
  \begin{nota} Let $d_Q:=\max\{ |u| \mid u$ is     a shortest string of its $\equiv_{state}$-equivalence class$\}$.
     %be an upper bound on the shortest strings in the  $\equiv_{state}$ classes.
   % Let $p_Q$ be an upper bound on the shortest strings in the  $\equiv_{state}$ classes.
 %  Let $d_Q:=2\times p_Q$.
   Then, $d_Q$ is an  upper bound on the length of periods $v$  and
  lags of $\om$-strings in $UP(Q)$.
%
  % \begin{nota} Let $d_Q:=\max\{ |u| \mid u$ is the shortest string in its $\equiv_{state}$-equivalence class$\}$.
%     %be an upper bound on the shortest strings in the  $\equiv_{state}$ classes.
%   % Let $p_Q$ be an upper bound on the shortest strings in the  $\equiv_{state}$ classes.
% %  Let $d_Q:=2\times p_Q$.
%   Then $d_Q$ is an  upper bounds on the length of periods $v$  and
%  lags $u$ of $\om$-strings in $uv^\omega\in UP(Q)$.
  % We denote by $p_Q$ and $d_Q$ an upper bound on the periods and
%  lags of $\om$-strings in $UP(Q)$.
\end{nota}
\begin{lem}\label{lem:later-choice}
  Assume $v\equiv_{state} vv$  and
   $u\equiv_{state} uv$ and $w=uv^\om =q_1q_2 \cdots$. Let $n$ and $m$ be the lengths of $u$ and $v$ respectively. Let $S_u$ be the set of states
   that appear in $u$. Then, for every $l\geq n$, the set of states that appear in $q_1 \cdots q_l$ is $S_u$.
   Moreover, for every $l,l'> n$ if $q_l=q_{l'}$, then $q_1 \cdots q_l\equiv_{state}q_1 \cdots q_{l'}$.
  % \begin{enumerate}
%     \item For every $l\geq m$, the set of states that appear in $q_1 \cdots q_l$ is $S_u$.
%     \item For every $i\geq 0$ and $l\geq m$: $q_1\cdots q_l\equiv_{state}q_1 \cdots q_{l+im}$.
%    % \item
%   \end{enumerate}
\end{lem}
\subsection{Restrictions on strategies of \Op}
\begin{defi}[Simple Strategies]
  A timed $\om$-word   $(a_0,t_0) (a_1,t_1) \cdots$ is simple timed  if there is $\Delta>0$ such that $t_i-t_0=\Delta\times i$  for all $i$.
%\end{defn}
 An \Op-strategy is \emph{simple} (for $\cA$) if it chooses only simple timed $UP_\cA(Q)$-sequences, i.e., the sequences of the form
$ (a_0,0)(a_1,\Delta) \cdots (a_i,\Delta\times i)\cdots$, where $a_0a_1\cdots \in UP_\cA(Q)$.
\end{defi}
%\marginpar{stutter free}
\begin{lem}[Reduction to simple strategies in $\cG_\cA$]
  \O has a winning strategy if and only if it has a simple winning strategy.
\end{lem}
\begin{proof}
Let $\sigma$ be a winning strategy for \O in $\cG_\cA$. We are going to construct a simple winning strategy.
First, for a timed $\om$-sequence $u:= (a_0,0) (a_1,t_1) \cdots$  let $first(u):=t_1$ be its first positive point and  $untimed(u)$
be  the corresponding untimed $\om$-sequence $a_0a_1\cdots$.

By $u^s$ we denote the simple timed sequence  $v$ such that $untimed(u)$ is equivalent to $untimed(v)\in UP(Q) $ and the scale of $v$ is  $t_1:=first(u)/d_Q$ where   $d_Q  $ is greater  than or equal to  the length of the lag of $untimed(v)$.

We are going to show that a  strategy $\tau$ which replaces $u$ moves of \O in $\sigma$ by the corresponding $u^s$ moves is winning.

Consider a play $\pi$  consistent  with $\tau$. We construct a play $
\pi^\sigma$  consistent with $\sigma$ such that
the states visited infinitely often in $\pi$ are the same   as the state visited infinitely often  in $\pi^\sigma$, 
and the  duration  of $\pi^\sigma$ is greater than or equal to   the duration  of $\pi$. Hence, if $\pi^\sigma$ is winning, then $\pi$ is winning.

For every prefix $\pi_n$ of $\pi$ where \O moves, we construct the corresponding prefix $\pi^\sigma_n$  of $\pi^\sigma$, where \O moves
and the following invariants hold:
\begin{description}
\item[I-last state]
$\pi_n$  and $\pi^\sigma_n$  end in the same states.
\item[I-recurrent states] The infix between $\pi_{n-1}$ and $\pi_n$ has the same states as the infix  between
$\pi^\sigma_{n-1}$  and $\pi^\sigma_n$. I.e., let $v_n$ (respectively, $v^\sigma_n$)
be such that $\pi_n=\pi_{n-1}v_n$ (respectively, $\pi^\sigma_n=\pi^\sigma_{n-1}v^\sigma_n$); then a state occurs in $v_n$ if and only if it occurs in  $v^\sigma_n$.  %, whenever  $n$-th interruption of    \I is inside a periodic part.
\item[I-duration] The duration of $\pi^\sigma_n$  is greater than or equal to the duration of $\pi_n$.
\end{description}

Assume that a prefix $\pi_n$ was constructed and  the corresponding prefix is $\pi^\sigma_n$
(they end in the same state). If $\sigma$ prescribes for  \O to choose $u$, let $u^s$ be the corresponding choice for $\tau$.
%
% I.e.,
%$u=u_1u_2 \cdots u_i \cdots$  and $u^s=v_1v_2^\om\in UP(Q)$ such that $u_1\equiv_{state} u_1u_2  \equiv_{state} v_1$ and
%$u_i\equiv_{state}u_j\equiv v_2 \equiv_{state} v_2v_2$ for $i,j>1$.

Now,  if \I accepts $u^s$, then we are done; both plays are accepting or rejecting.
If \I interrupts at a state $q$, we consider two cases.

A) \I interrupted $u^s$ at time  $t$  at a position inside the period. Then, we can find $t'>t$ such that $u$ is at the state  $q$ at $t'$  and the corresponding prefixes are $\equiv_{state}$-equivalent. In fact, by  Lemma \ref{lem:up} and Lemma \Ref{lem:later-choice}, there are infinitely many such positions in $u$.  Now, we constructed $\pi^\sigma_{n+1}$ prefix with the duration greater  than the duration of corresponding $\pi_{n+1}$ prefix.

B) \I interrupts $u^s$ at time $t$  in a  state $q $  at a position in the   lag; hence, $t<t_1=\mathit{first}(u)$.  We choose the corresponding interruption point in   $u$. %, i.e., $q_1$ at this point and before it the set of state $Q_1$ appear.
It is at time   greater than or equal to $ t_1$. Hence, the invariants are preserved.

%\I interrupts at time $t$  at state $q_1 $ in the  prefix such that $Q_1$ is the set of states that appear before $q_1$, hence $t<t_1=first(u)$.  We choose corresponding interruption point in the prefix  $u$, i.e., $q_1$ at this point and before it the set of state $Q_1$ appear. It is at distance  $\geq t_1$. Hence, the invariants are preserved.

It remains to show that $\pi$ is winning for \O  if $\pi^\sigma$ is winning for \Op.

If the plays are finite, then $\pi$ is winning iff $\pi^\sigma$ is winning.

If the plays are infinite, then  (1) by the recurrent state invariant the  same states  appear infinitely often in $\pi$ and in $\pi^\sigma$,
and (2) if  \I causes   time to converge in $\pi^\sigma$  to $t^\sigma<\infty$, then in $\pi$ time converges to a  value bounded by $ t^\sigma<\infty$. Since  \O wins $\pi^\sigma$, we obtain that  \O wins $\pi$.
%
%Assume that $\pi^\sigma$ is winning because the interruptions of  \I  cause that time converges to $t^\sigma<\infty$.
%Then in $\pi$ time converges to the value $\leq t^\sigma<\infty$. Hence, \O wins $\pi$.
%
%If times diverges, then \I infinitely often interrupts inside period.
%%
%%We are going to show that in this case the same states  appear infinitely often in $\pi$ and in $\pi^\sigma$.
%% Let $s_1,  \cdots, s_k$ be the $\omega$-strings chosen infinitely often. From some point on only these moves are chosen.
%% If
%%
%
%
%By the recurrent state invariant the  same states  appear infinitely often in $\pi$ and in $\pi^\sigma$.
%Assume that $\pi^\sigma$ is winning because the interruptions of  \I  cause that time converges to $t^\sigma<\infty$.
%Then in $\pi$ time converges to the value $\leq t^\sigma<\infty$. Hence, \O wins $\pi$.
%
%If times diverges, then \I infinitely often interrupts inside period.
%%
%%We are going to show that in this case the same states  appear infinitely often in $\pi$ and in $\pi^\sigma$.
%% Let $s_1,  \cdots, s_k$ be the $\omega$-strings chosen infinitely often. From some point on only these moves are chosen.
%% If
%%
%
%
%By the recurrent state invariant the  same states  appear infinitely often in $\pi$ and in $\pi^\sigma$.
\end{proof}

Now, we will play explicitly on a   finite arena. This game corresponds to $\cG_\cA$-game when \O is restricted to use only simple strategies.

% A strategy is called almost positional, if at every vertex it always chooses the same untimed move (timed scales might be different).

Let $\cA:=(Q,\Sigma_{in},\Sigma_{out}, \delta, q_{init},pr)$ be an automaton.
We construct the following arena:

\bigskip
%\fbox{
%\begin{minipage}{0.9\linewidth}
%
%\begin{center}
%  Arena for simple strategies game
%\end{center}
%
%
%\textbf{States:} $\{q_{init} \}\cup Q\times \Sigma_{in} \cup  Q\times \Sigma_{in}\times UP(Q)$, where $Q_{\mathcal I}:=\{q_{init}\}\cup  Q\times \Sigma_{in}\times UP(Q)$ is  the set of \I nodes,  and
%$Q_{\mathcal O}:= Q\times \Sigma_{in}$ is  the set of  \O nodes. The set $F$  of final nodes  is $\{(q,a,u)\mid $ the maximal priority that appears infinitely often in $u$ is
%even$\}$.
%
%\textbf{Edges:}
%\begin{enumerate}
%  \item $q_{init} $ to $(q_{init},a)$ for every $a\in \Sigma_{in}$.
%  \item $(q,a, u) $ to $(q',b)$ for $u\in UP(Q)$,  $q'$ a state in $ u$ and  $a\neq b\in \Sigma_{in}$. This edge will have  priorities (maybe several). $p$ is a priority of this edge if  there is $n$ such that $u(n)=q'$ and the maximal priority among   $u(1),\dots , u(n)$ is $p$;
%      if the preperiodic part of $u$  has length $<n$ then  the priority $p$ is called  big, otherwise it is called small.
%
%
%  \item $(q_0,a)$ to  $(q_0,a,u)$, where   $ u=(q_1, q_2,\dots)\in UP(Q)$  such that $\exists b_i\in \Sigma_{out}. ~q_{i+1}=\delta(q_i,(a,b_i))$ for all $i$.
%\end{enumerate}
%%$q_{init} $ to $(q_{init},a)$ for every $a\in \Sigma_{in}$
%%
%%$(q,a, u) $ to $(q',b)$ for $u\in UP(Q)$,  $q'$ a state in $ u$ and  $a\neq b\in \Sigma_{in}$.
%%
%% $(q_0,a)$ to  $(q_0,a,u)$, where   $ u=(q_1, q_2,\dots)\in UP(Q)$  such that $\exists b_i\in \Sigma_{out}. ~q_{i+1}=\delta(q_i,(a,b_i))$ for all $i$.
%
%\end{minipage}
%}

\fbox{
%\begin{center}
%  Timed game on finite arena
%\end{center}

\begin{minipage}{0.9\linewidth}

\begin{center}
  Arena $\cG^{arena}_\cA$  
\end{center}

\textbf{Nodes:} $\{q_{\new} \}\cup Q\times \Sigma_{in} \cup  Q\times \Sigma_{in}\times UP(Q)$, where $Q_{\mathcal I}:=\{q_{\new}\}\cup  Q\times \Sigma_{in}\times UP(Q)$ is  the set of \I nodes,  and
$Q_{\mathcal O}:= Q\times \Sigma_{in}$ is  the set of  \O nodes. The set $F$  of final nodes  is $\{(q,a,u)\mid $ the maximal priority that appears infinitely often in $u$ is
even$\}$.

\textbf{Edges:}
\begin{enumerate}
  \item $ {q}_{\new} $ to $(q_{init},a)$ for every $a\in \Sigma_{in}$.
   \item $(q_0,a)$ to  $(q_0,a,u)$, where   $ u=(q_1, q_2,\dots)\in UP(Q)$  such that $\exists b_i\in \Sigma_{out}. ~q_{i+1}=\delta(q_i,(a,b_i))$ for all $i$.
  \item There is an edge from $(q,a, u) $ to $(q',b)$ for $u\in UP(Q)$,  % $q'$ a state in $ u$ and  $a\neq b\in %\Sigma_{in}$,
       if $a\neq b\in \Sigma_{in}$ and 
  there is $n$ such that $u(n)=q'$; this edge will be labeled  by $p$ - the maximal priority among   $u(1),\dots , u(n)$  and it  will be called small if   $u(1),\dots , u(n)$  is in the lag part of $u$; otherwise, it will be called big.

 % This edge will have  priorities (maybe several). $p$ is a priority of this edge if  there is $n$ such that $u(n)=q'$ and the maximal priority among   $u(1),\dots , u(n)$ is $p$;
%      if the lag  part of $u$  has length $<n$ then  the priority $p$ is called  big, otherwise it is called small.
%
 %% \item $(q_0,a)$ to  $(q_0,a,u)$, where   $ u=(q_1, q_2,\dots)\in UP(Q)$  such that $\exists b_i\in \Sigma_{out}. ~q_{i+1}=\delta(q_i,(a,b_i))$ for all $i$.
\end{enumerate}
%$q_{init} $ to $(q_{init},a)$ for every $a\in \Sigma_{in}$
%
%$(q,a, u) $ to $(q',b)$ for $u\in UP(Q)$,  $q'$ a state in $ u$ and  $a\neq b\in \Sigma_{in}$.
%
% $(q_0,a)$ to  $(q_0,a,u)$, where   $ u=(q_1, q_2,\dots)\in UP(Q)$  such that $\exists b_i\in \Sigma_{out}. ~q_{i+1}=\delta(q_i,(a,b_i))$ for all $i$.
%
\end{minipage}
}
\label{arena} 

\bigskip
This arena is a finite graph (in   (3)   construction of edges depends on $n\in \nat$, however,  due to   periodicity of $u$, the number
of edges is finite and they are computable).  
 %If there are $k$ priorities, then out-degree of nodes of the form $(q,a, u) $ is at most $2k$;
 There might be  parallel edges from $(q,a,u)$ to $(q',b)$  which will have labels  $(p,big)$ or $(p,small)$ for a priority $p$.
 If $k$ is the number of priorities, then from $(q,a,u)$ to $(q',b)$  for every $p\leq k$   there is at most
 one edge $(p,small)$;  and   by Lemma \ref{lem:later-choice},  there is a unique $p'$  such that there is an edge $(p',big)$
 from  $(q,a,u)$ to $(q',b)$; moreover, this $p'$ is the maximal of the priorities assigned to the states of $u$.
 % parallel edges of.

\bigskip
%Now consider the following game on a (finite)  arena.
Now, consider the following game.

\fbox{

\begin{minipage}{0.9\linewidth}

\begin{center}
  Timed game $\cG^{arena}_\cA$ on our finite arena
\end{center}

Initial step. \I  chooses $a\in \Sigma_{in}$ and the play  proceeds to $(q_\mathit{init},a)$;  $t_0$ is set to 0.

From a  state $(q,a)$ at the $i$-th move at time $t_i$ player  \O chooses an edge to $(q,a,u)$ and a time scale $\Delta_i$. 
It defines a  sequence $Y_i$ on $[t_i,\infty)$ as
$Y(t_i+\Delta_i\times n)=q$ if $q$ is the $n$-th state of $u$.

From  a state $ (q,a, u)$,  at $i+1$-move Player \I either (A) accepts and the play  ends,   %It is winning for \O iff $ (q,a ,u)\in F$,
 or (B)
interrupts at $t_{i+1}>t_i$ by chosen $b\neq a$. Then, the  play moves to $(q_n,b) $ by an edge with priority $p$  if $t_i+\Delta_i\times (n-1)<t_{i+1}\leq t_i+\Delta_i\times n$ and $p$ is the maximal of priorities among $q_1,\cdots, q_n$.

  \emph{Winning Condition.}
  \O wins a play $\pi$  if either $\pi$ is finite (this happens if \I has chosen option (A)) and its last state is in $F$, or $\pi$ is infinite and $\lim t_i$ is finite, or
  the maximal  priority that appears infinitely often is even.
  \end{minipage}
  }

  \bigskip
  Note that in the above game the arena is finite, yet each player has uncountable many moves,  when he/she chooses time.
  The following is immediate.
  \begin{lem}
    \O wins $\cG^{arena}_\cA$ if and only if \O has a winning   simple strategy for $\cG_\cA$-game.
  \end{lem}
Next, we are going to restrict further the strategies of  player \O in $ \cG^{arena}_\cA$.

 A strategy is called \textbf{almost positional}, if at every vertex of the arena it always chooses the same untimed move (time scales might be different).
%\marginpar{where to state determinacy}
 \begin{lem}[Almost positional winning strategy in $\cG^{arena}_\cA$]

   \O has a simple winning strategy iff \O has an almost positional simple winning strategy.
 \end{lem}
\begin{proof}
First, note that  $\cG^{arena}_\cA$ game  is determined, i.e., one of the player has a winning strategy.

Let $\sigma$ be a simple strategy for \Op.
For every vertex $s$ in  $Q_{\mathcal O}$,
%let $out^\sigma_s$ be the cardinality of the set of possible (untimed) moves from $s$  which are  consistent with $\sigma$.
let $\mathit{Consistent} (\sigma,s)$ be   the set of possible (untimed) moves from $s$  which are  consistent with $\sigma$, i.e., $u\in\mathit{Consistent} (\sigma,s)$ if 
there is a (finite) run
$\pi:=s_0s_1\dots s_n$ consistent with $\sigma$  such that $s=s_n$  and  $\sigma$ prescribes to choose $u'$ after $\pi$, where
$u=untimed(u')$. 
Let $out^\sigma_s$ be the cardinality of $\mathit{Consistent} (\sigma,s)$. 

Assume $\sigma$ is winning.  We prove  by  induction on $\sum out^\sigma_s$ that \O has an almost positional simple winning strategy..

Basis.    $out^\sigma_s=1$  for every $s$.  Hence,  $\sigma$ is almost positional.

Inductive step. Assume $out^\sigma_s>1$. Let $u$ be a move from $s$ consistent with $\sigma$.

%
%We consider the arena $G_u$ where only $u$ moves (with different scales) are available from $s$ and the arena  $G_{\neg u}$ where $u$ moves are forbidden and only the   moves from $\mathit{Consistent} (\sigma,s)\setminus \{u\}$ are   allowed from $s$. On other nodes and edges 

Let arena  $G_u$ be obtained from $\cG^{arena}_\cA$  when only $u$ moves (with different scales) are available from $s$,
 and on the other nodes $\cG^{arena}_\cA$  and $G_u$ 
coincide. 
Similarly, let  arena  $G_{\neg u}$ be obtained from $\cG^{arena}_\cA$   when $u$ moves are forbidden and only the   moves from $\mathit{Consistent} (\sigma,s)\setminus \{u\}$ are   allowed from $s$,   and on the other nodes $\cG^{arena}_\cA$  and $G_u$ 
coincide.

If \O wins in $G_u$ or in $G_{\neg u}$,   by the inductive hypothesis \O  has  an   almost positional winning strategy in $G_u$ or in $G_{\neg u}$, 
and this strategy will be winning in $\cG^{arena}_\cA$. 

Assume  \I wins in these games by strategies $\mu_u$ and $\mu_{\neg u}$.
We are going to show that this implies that \I has a winning strategy in $\cG^{arena}_\cA$  and this  contradicts that
$\sigma$ is a winning strategy for \Op.

Let $\pi v$ be a finite  play in $G$ from $s$.
We say that a move $v$ is a $ G_u$ move if the last move from $s$ in $\pi v$ was $u$, otherwise $v$ is a $G_{\neg u}$ move.
%
%Partition it into a play in $ G_u$ and a play in $G_{\neg u}$.
%A move $v$  is in $G_u$ iff  the last move from $s$ before  $v$ is $u$.

Let $View_u(\pi)$ (respectively, $View_{\neg u}(\pi)$) be  the sequence of   $ G_u$ (respectively,  $G_{\neg u}$) moves in $\pi$.

Define a strategy $\mu$ for \I as follows. Assume $\pi$ ends  in a position for \Ip, then
$$
\mu(\pi)  =\begin{cases}
			\mu_u(View_u(\pi) ) & \text{if the last move of $\pi$ is a $ G_u$ move}\\
            \mu_{\neg u} (View_{\neg u}(\pi) )& \text{otherwise}
		 \end{cases}
$$

  Consider a play $\pi$  which is consistent with $\mu$. 
  If from some moment on $\pi$ stays in $ G_u$ (respectively,  $G_{\neg u}$), then from this moment it is consistent with 
  $\mu_u$ (respectively, with $\mu_{\neg u}$), and therefore it is winning for \Ip.
  
 Assume $\pi$ infinitely often  switches between  $ G_u$  and $G_{\neg u}$.
Since  $View_u(\pi)$ (respectively, $View_{\neg u}(\pi)$) is consistent with $\mu_u$ (respectively, with $ \mu_{\neg u} $) we obtain that 
  $View_u(\pi)$   and $View_{\neg u}(\pi)$ are winning plays for \Ip. Hence, their durations are infinite.
 The maximal priority that occurs infinitely often in  $View_u(\pi)$  is odd and the maximal priority that occurs infinitely often  in  $View_{\neg u}(\pi)$   is odd. Hence, the maximal priority that occurs infinitely often in $\pi$ is odd, and $\pi$ is winning play for \Ip.
   %such that  if \I moves after  a subplay $\pi_n$ and the last move from $s$ in $\pi_n$ was $u$, then \I plays according to $\tau_u$ 
%  in $View_u(\pi_n)$. Otherwise,  the last move from $s$ in $\pi_n$ was not $u$, then player \I plays according to  $\tau_{\neg u}$  
%  in  $View_{\neg u}(\pi_n)$).
%  
%   Recall that  $\tau_u$  and  $\tau_{\neg u}$  are winning strategies for \Ip. If from some time on 
%   we make only $u$ move from $s$, then from   the duration  of this play is unbounded,
%  and the parity is good for Player \Ip. Hence, $\pi$ is winning for \I - contradiction.
  \end{proof}

Finally, we restrict/eliminate  the time scale that  Player \O chooses.
   \begin{lem}[$2^{-i}$ scale]
    \O has an almost positional simple winning strategy if and only if
    \O has an almost positional simple winning strategy when the  $i$-th time move has the scale $2^{-i}$.

 \end{lem}
  \begin{proof}
    Let $\sigma$ be an almost positional simple winning strategy for \O and
    $\sigma^{exp}$ be the corresponding  almost positional simple   strategy for \O when  the  $i$-th time move has  scale~$2^{-i}$.

    Assume $\pi^{exp}$ is a play consistent with $\sigma^{exp}$ which is winning for \I.
    Toward a contradiction, we are going to construct $\pi$ consistent with $\sigma $ which is winning for \Ip. % - contradiction.

    We will keep the following invariant:
    After the $i$-th move of \I the corresponding subplays  $\pi_i$ and $\pi^{exp}_i$  (1) have the same last state, (2)  pass through the same set of states, and  (3) infinitely often the duration  of  $\pi_i$   is greater than 
    the duration  of   $\pi^{exp}_i$. Moreover,  (4) if $\pi_{i+1}=\pi_iw_{i+1}$ and $\pi^{exp}_{i+1}=\pi^{exp}_iw'_{i+1}$, then $w_{i+1}$ is state equivalent
    to $w'_{i+1}$.

    Let $u^{exp}$ be the $i+1$-th  move of \Op. Assume \I interrupts it at a state $q$ at time $t$. If $q$ is inside the period,
    we can find $q$ which   at time $t'>t$ (far enough in the $u$ move) and interrupts $u$ at the same $q$ with the same past.
    If $q$ is inside the prefix lag, we pick up the corresponding $q$ in the prefix of $u$. Hence, $\pi_{i+1}=\pi_iw_{i+1}$ and $\pi^{exp}_{i+1}=\pi^{exp}_iw'_{i+1}$ and $ w_{i+1}$ is state equivalent
    to $w'_{i+1}$.

%    It is clear that the first two invariants are preserved. 
It is clear that invariants (1), (2) and (4)  are preserved. 
    If the duration of $\pi^{exp}$ is infinite, then
    there are infinitely many interruptions inside the periods of \O moves. Hence, for each such interruption at $n$-th move
    we can choose the corresponding  interruption in $u$ that ensures that the duration of $\pi_n$ is greater than the duration of $\pi^{exp}_n$.  
    We assumed that $\pi^{exp}$ is winning for \Ip. Therefore, the duration of $\pi^{exp}$ is infinite and   the third property  holds.

    To sum up, we proved that if there is a play consistent with $\sigma^{exp}$  which is winning for \Ip,
    then there is a play consistent with $\sigma $  which is winning for \I a  contradiction to the assumption that
    $\sigma $  is a winning strategy for \Op.
  \end{proof}
\subsection{Decidability}
We proved that \O has a winning strategy iff
 \O has an almost positional simple winning strategy where  the $i$-th time move has the scale $2^{-i}$.
 There are finitely many such strategies.
 Hence, to decide whether \O has a winning strategy it is enough to check
 whether such a strategy is winning for \Op.

 Let $\sigma$ be an almost positional  strategy for \O with the time scale $2^{-i}$ on the $i$-th move.
It chooses exactly one edge from every vertex $(q,a)$.
Consider the sub-graph $G_\sigma$ of $\cG^{arena}_\cA$  where from $(q,a) $ only the edge specified by $\sigma$ is chosen (all other edges from $(q,a)$ are removed).
This is an arena for one player \Ip.
% We slightly modify its edges, in order to avoid multiple priorities on edges.
Recall  that the edges from the nodes of the form $(q,a,u)$ are labeled (see the definition of the  arena on page \pageref{arena}).
%Let $u=q_1\cdots q_n\cdots$. 
The edge which corresponds to the move from $(q,a,u) $ when  \I interrupts at $q_n$  is  assigned a priority $p$, where $p$ is 
  the maximal   priority   assigned to the states $q_1,\cdots, q_n$; this edge   is small % (respectively, big)
if $q_1 \cdots q_n$ is   inside lag  of $u$; otherwise, it is big. The edge   
% Hence, if there are $k$ priorities then from every node of the form $(q,a,u)$   exit at most $2k$ edges and each edge
will carry a label $(p,small)$   or $(p,big)$.  
Note ``big'' and ``small'' are just names that indicate whether an interruption  appears
in the period part or in the lag part, and they do not
affect the order of priorities.

 %If an edge is labeled by $k$  priorities $p_1,\dots ,p_k$ we replace it by $k$ parallel edges
%labeled by $p_1$ to $p_k$ respectively. Moreover, if $p_i$ is both big and small priority, we replace an edge labeled by $p_i$
%by two parallel edges: one labeled by $p_i$ big, and the other labelled by $p_i$ small (big and small are just names that indicate whether a label appears
%in the period part or in the lag part, and they do not
%affect the order).
% Hence, if there are $k$ priorities then from every node of the form $(q,a,u)$ will exit at most $2k$ edges and each edge
%will carry label $(p,big)$ or $(p,small)$ for $p\leq k$.
%%
%Construct an untimed arena for every positional quasiperiodic strategy $\sigma$ (time is ignored).
%
%There is a short edge from $s$ to $s'$ labeled by a set $T\subseteq Q$ if $\sigma(s)=u$ there is a prefix of $u$ (which ends before the period) labeled by $T$ which leads  to $s'$.
%
%There is a long edge from $s$ to $s'$ labeled by a set $T\subseteq Q$ if $\sigma(s)=u$ there is a prefix of $u$  which ends inside  the period and it  is  labeled by $T$ and  leads  to $s'$.

\begin{lem}[Deciding whether a strategy is winning] \label{lem:decide-strategy-iswinning} Let $\sigma$ be an almost positional winning strategy for \O with the time scale $2^{-i}$ on the $i$-th move.
Player \I wins against $\sigma$ if and only if in $G_\sigma$  either (A) there is a reachable  node $(q,a,u)\not \in F$  or (B)
there is a reachable cycle $C$  with at least one  edge labeled by a big priority and
the maximal priority  in $C$ is odd. % edge with the  priority labeling good for \I.
\end{lem}
\begin{proof}
  It is clear that if (A) holds, then \I wins against $\sigma$. She only needs to reach such a  node $(q,a,u)\not \in F$ and then stop playing.
If (B) holds, then \I can reach such a cycle $C$ and then move along its edges. Moreover, when she  goes over an edge labeled by a big priority $p$,
she can choose a time $t$ which is at distance at least one from the time when the corresponding edge to $(q,a,u)$ is taken (by Lemma \ref{lem:later-choice}, she can choose times as big as she wishes).
This will ensure that every round along $C$ takes at  least one time unit; therefore, $\lim t_i=\infty$. Since the maximal priority on $C$ is odd, this play is winning for \Ip.

Now, we assume that \I wins and show that (A) or (B) holds.
If (A) does not hold, then a winning play for \I must be infinite.
We are going to find a cycle  as described in (B).
If  only small priorities appear infinitely often, then
the duration of the play is finite and it cannot be winning for \Ip.
Indeed, assume that after  $t_j$  only small priorities appear in the play. Recall that $d_Q$ is an upper bound on the length of lags
for $UP(Q)$ moves. Then, $t_{i+1}-t_i< d_Q\times 2^{-i}$ for all $i\geq j$.
Hence, $\lim t_i \leq t_j+2d_Q<\infty$ and \I loses.

Therefore, there is an edge  $e_b$  with a big priority that appears infinitely often.
Let $p$ be the maximal priority that appears infinitely often (it is odd). Then, some edge $e_p$ with $p$ appears infinitely often,
and there is a sub-play that starts and ends at $e_p$ and passes through $e_b$. This sub-play is a cycle that passes through an edge with a big
priority and the maximal priority of its edges is    $p$ and $p$ is odd.
%
% If  $\pi$ is a winning play, then it contains an edge $e$ with such that the edge appears in $\pi$ infinitely often and its priority $p$ is odd and the  maximal between priorities that appears infinitely often in $\pi$.
%If from some point  on all edges in $\pi$ have a small priority, then $t_{i+1}-t_i< m_Q\times 2^{-i}$ for all $i$ after this point. Therefore,
%$\lim t_i <\infty$ and this contradicts that $\pi$ is a winning play for $I$.
\end{proof}

There are only finitely many   almost positional   strategies for \O with the time scale $2^{-i}$ on the $i$-th move (in $\cG^{arena}_\cA$), and  by Lemma
\ref{lem:decide-strategy-iswinning},  it is decidable  whether such a  strategy is winning,  and \O wins  iff it has such a winning strategy. Hence, we obtain:
\begin{lem}
  It is decidable whether \O has a wining strategy $\cG^{arena}_\cA$.
\end{lem}
Recall  \O has a winning strategy in $\cG^{arena}_\cA$ iff \O has a winning strategy in $\cG_\cA$ iff \O has a winning strategy in $\cG_\Psi$ iff
there is a causal operator that implements $\Psi$.  Therefore, we obtain:
\begin{thm}[Decidability of Synthesis for RC Signals]\label{th:dec-synth-rc}
  The Church synthesis problem over RC signals  is decidable.
\end{thm}
\section{The Church Synthesis Problem over FV Signals is Decidable}\label{sect:ch-fv}
In this section a proof of the decidability of Church's synthesis problem over FV signals is outlined.
\synth*
The proof is very similar to the proof of Theorem \ref{th:dec-synth-rc}, and does not require  novel insights or techniques. Here, we explain   the differences.

We start with a game.
Consider the following two-player game $\cG_\Psi$.  During the play  \I constructs a (timed $\om$-sequence for)  input signal  $X$ and \O constructs  an output $Y$.
 
Round 0. Set $t_0:=0$.
\begin{enumerate}\item

 \I chooses a value    $a_0\in \Sigma_{in}$  which holds at $t_0$.  % and for a while after  $t_0$.
 \item  \O responds  with $b_0\in \Sigma_{out}$.
 \item \I chooses $a^d_0\in \Sigma_{in}$ to hold for a while after $t_0$  (here and below the superscript
 $d$ in $b^d$ indicates that the duration of  $b$   is non-singular, i.e., this input $b$ lasts for   a positive amount of time).
           \item \O responds with a timed $\omega$-sequence (equivalently, with a signal) $Y_0$ on $(t_0,\infty)$.

      Now, \I has two options:

      A) \I accepts. In this case, the play ends. \I has  constructed $X$ which is  $a_0$ at $t_0=0$ and $a^d_0$ at  all $t> 0$, and \O has constructed $Y$ which is $b_0$ at $t_0$ and $Y_0$   -  the signal chosen by \O in $(t_0,\infty)$.

 B)  \I interrupts at some $t_1>t_0$ by choosing left or right discontinuity in the input.
  \begin{description}
      \item[Discontinuity from the left] \I chooses $a_1\neq a^d_0$  and the game proceeds  to step (1) of round 1, or  %\I  sets $a_1= a^d_0$  and 
      \item[Discontinuity from the right]  \I  sets $a_1= a^d_0$  and   $a^d_1\neq  a^d_0$ and the game proceeds  to step (3) of round~1.
      \end{description}

    %  B)  \I interrupts at some $t_1>t_0$ by choosing
%    %  \begin{enumerate}
%     % \item
%     $a_1\neq a_0$ which holds at $t_1$ and for a while after  $t_1$.  % and going to step 1 of round 1.
%      %\item $a^d_1\neq  a^d_0$ and going to step 3 of round 1.
%      %\end{enumerate}
%      $X'_0$ constructed by \I is $a_0$   in  $[t_0,t_1)$  and is $a_1$ at $t_1$,
%       %while in (b) $a^d_0$ at $t_1$ and $a^d_1$ for a while after $t_1$;
%        $Y'_0$ constructed by \O  is  the restriction of $Y_0$
%        on the interval $[t_0,t_1)$. Now,  proceed  to step 1 of round 1.
%
\end{enumerate}
Round $i+1$.
\begin{enumerate}
\item
\O chooses $b_{i+1}\in \Sigma_{out}$ to hold at $t_{i+1}$.
\item
\I
 chooses $a^d_{i+1}$ to hold for     a positive amount of time  after $t_{i+1}$.
\item
  \O responds with a timed $\omega$ sequence $Y'_{i+1}$  (equivalently, with a signal) on $(t_{i+1},\infty)$.
\item Now, \I has two options.

% A) \I accepts. In this case the play ends. The signal $X$ constructed by \I coincides with  $X'_i$ on $[0,t_{i+1})$ and equals $a_{i+1}$ at $t_{i+1}$ and equals  $a^d_{i+1}$ on   $(t_{i+1},\infty)$. The signal $Y$ constructed by \O coincides with $Y'_j$ on $(t_j,t_{j+1})$ equals $b_{i+1} $ at $t_{i+1}$ and coincides  with $Y'_{i+1}$ on $[t_{i+1},\infty)$.

A) \I accepts. In this case the play ends. The signal $X$ constructed by \I 
is $a_j$ at $t_j$, is $a^d_j$ on $[t_j,t_{j+1}) $ for $j\leq i$,  is $a_{i+1} $ at $t_{i+1}$ and is equal $a^d_{i+1}$ on  $(t_{i+1},\infty)$.
%and coincides with  $X'_i$ on $[0,t_{i+1})$ and equals $a_{i+1}$ at $t_{i+1}$ and equals  $a^d_{i+1}$ on   $(t_{i+1},\infty)$. 
The signal $Y$ constructed by \O is defined similarly, by concatenating the signals constructed by \O at rounds $0,\dots, i+1$.
%coincides with $Y'_j$ on $(t_j,t_{j+1})$ equals $b_{i+1} $ at $t_{i+1}$ and coincides  with $Y'_{i+1}$ on $[t_{i+1},\infty)$.

       B) \I  interrupt at some $t_{i+2}>t_{i+1}$ by choosing:
      \begin{enumerate}
      \item $a_{i+2}\neq a^d_{i+1}$  in this case round $i+2$ starts at step (1) (in this case $X$ is left discontinuous at $t_{i+1}$) or by
      \item    $a_{i+2}=a^d_{i+1}$  and    $a^d_{i+2}\neq  a^d_{i+1}$. In this case round $i+2$ starts at step (3)  (in this case $X$ is right discontinuous at $t_{i+2}$).
      \end{enumerate}
%\end{enumerate}
\end{enumerate}  

  The  signal $X$ constructed by \I  corresponds to the timed sequence $(a_0, a^d_0,0)(a_1,a^d_1, t_1) \dots $ $(a_i,a^d_i,t_i)\dots $,
     and  the signal $Y$   constructed by  \O is the concatenatin of the signals constructed by \O in each round..

\emph{Winning conditions.}     \O wins the play iff $\Psi(X,Y)$ holds or \I interrupts infinitely often and
the duration of the play is finite, i.e., $\lim t_n< \infty$, where $t_n$ is the time of $n$-th interruption.

Similarly, to Lemma \ref{lem:red}, we have
 \begin{lem}[Reduction to games] \label{lem:red-fv}
    \O wins $\cG_\Psi$  if and only if there is a C-operator that implements $\Psi$.
  \end{lem}
  Note that in this game \O has two types of moves.  In step (1) she chooses a letter in $\Sigma_{out}$ and in Step (3) she chooses a $\Sigma_{out}$-signal over open interval $(t,\infty)$.
 Similarly  to the proof of Theorem \ref{th:dec-synth-rc}, we are going to restrict the \O moves in step (3).
The first restriction is a time schedule.
When  \O chooses a timed sequence (equivalently a signal), we require that the distance between consecutive time points in the sequence  is the same.
We call such a sequence  \emph{simple timed}.
Another restriction is on the sequence of values.
We require that this sequence is ultimately periodic.
But first, we consider a game over an automaton in which \O chooses signals over states
instead of signals over $\Sigma_{out}$. She chooses   the run which  corresponds to a choice of an input  over $\Sigma_{out}$.

Consider the game $G_\cA$  over a deterministic parity automaton $\cA:=  (Q,\Sigma_{in}\times \Sigma_{out} ,\delta,$\linebreak$q_{init}, pr)$  equivalent to $  \Psi^D$.
  In this game \O
will choose % timed $\om$-sequences of states, instead of $\om$-sequences of $\Sigma_{out}$.
signals over $Q$ instead of signals over  $\Sigma_{out}$.

 \begin{center}
   Game $G_\cA$.
 \end{center}
Round 0. Set $t_0:=0$.
\begin{enumerate}\item

 \I chooses a value    $a_0\in \Sigma_{in}$  which holds at $t_0$.  % and for a while after  $t_0$.
 \item  \O responds  with $q_0:=\delta(q_{init},(a_0,b_0))$ for some $b_0\in \Sigma_{out}$.
 \item \I chooses $a^d_0\in \Sigma_{in}$ to hold for a positive amount of time  after $t_0$.
           \item \O responds with a  signal $Y_0$ on $(t_0,\infty)$ over $Q $ such that $Y_0$ is a run of $\cA$  (from $q_0$) when the signal over $\Sigma_{in}$ is fixed to be  $a^d_0$.

      Now, \I has two options.

      A) \I accepts. In this case, the play ends. \I has constructed $X$ which is  $a_0$ at $t_0=0$ and $a^d_0$ at  all $t> 0$, and \O has constructed $Y$ which is $q_0$ at $t_0$ and is $Y_0$   -  the signal chosen by \O in $(t_0,\infty)$.

 B)  \I interrupts at some $t_1>t_0$ by choosing left or right discontinuity in the input.
Set $q_1$ to be the left limit of $Y_0$ at $t_1$. 
  \begin{description}
      \item[Discontinuity from the left] \I chooses $a_1\neq a^d_0$ % and set $q_i$ to be the left limit of $Y_0$ at $t_i$, 
           and the game proceeds  to step (1) of round 1, or  %\I  sets $a_1= a^d_0$  and 
      \item[Discontinuity from the right]  \I  sets $a_1= a^d_0$  and  $a^d_1\neq  a^d_0$ and the game proceeds  to step (3) of round~1.
      \end{description}

    %  B)  \I interrupts at some $t_1>t_0$ by choosing
%    %  \begin{enumerate}
%     % \item
%     $a_1\neq a_0$ which holds at $t_1$ and for a while after  $t_1$.  % and going to step 1 of round 1.
%      %\item $a^d_1\neq  a^d_0$ and going to step 3 of round 1.
%      %\end{enumerate}
%      $X'_0$ constructed by \I is $a_0$   in  $[t_0,t_1)$  and is $a_1$ at $t_1$,
%       %while in (b) $a^d_0$ at $t_1$ and $a^d_1$ for a while after $t_1$;
%        $Y'_0$ constructed by \O  is  the restriction of $Y_0$
%        on the interval $[t_0,t_1)$. Now,  proceed  to step 1 of round 1.
%
\end{enumerate}
Round $i+1$.
\begin{enumerate}
\item
\O updates  $q_{i+1} :=\delta(q_{i+1},(a_{i+1},b_{i+1}))$ for some $b_{i+1}\in \Sigma_{out}$ to hold at $t_{i+1}$.
\item
\I
 chooses $a^d_{i+1}$ to hold for a positive amount of time  after $t_{i+1}$.
\item
  \O responds  with a signal $Y_{i+1}$ over $Q$ on $(t_{i+1},\infty)$
  which is a run of $\cA$ from $q_{i+1}$ when the signal over $\Sigma_{in}$ is fixed to be $a^d_{i+1}$.
\item Now, \I has two options:

 A) \I accept. In this case the play ends. 
 %The signal $X$ constructed by \I coincides with  $X'_i$ on $[0,t_{i+1})$ and equal to $a_{i+1}$ at $t_{i+1}$ and equal to $a^d_{i+1}$ in   $(t_{i+1},\infty)$.
  The signal $Y$ constructed by \O coincides with $Y_l$ on $(t_l,t_{l+1})$,  is equal to $q_{l+1} $ at $t_{l+1}$ (for $l\leq i$),  and coincides  with $Y_{i+1}$ on $(t_{i+1},\infty)$.

       B) \I  interrupt at some $t_{i+2}>t_{i+1}$. Set $q_{i+2}$ to be the left limit of $Y_{i+1}$ at $t_{i+2}$. \I chooses either
      \begin{enumerate}
      \item $a_{i+2}\neq a^d_{i+1}$  and then round $i+2$ starts at step (1) (in this case $X$ is left discontinuous at $t_{i+2}$) or by
      \item    $a_{i+2}=a^d_{i+1}$  and    $a^d_{i+2}\neq  a^d_{i+1}$. In this case round $i+2$ starts at step (3)   (in this case $X$ is right discontinuous at $t_{i+2}$).
      \end{enumerate}
%\end{enumerate}
\end{enumerate}  
\emph{Winning conditions.}  \O wins if either the sequence of states produced by \O is accepted by $\cA$, or there were infinitely many moves, but $\lim t_i$ is finite.

Similarly, to Lemma \ref{lem:eq-aut-game}, we have
\begin{lem}[$\cG_\psi$ and $\cG_\cA$ are equivalent]
  \O wins in $\cG_\Psi$ if and only if \O wins in $\cG_\cA$.
\end{lem}

Exactly as in  Subsection \ref{sub:ult-periodic} we define $\equiv_{state}$-equivalence over finite strings and 
an equivalence over $\omega$-strings. We  define a finite set of ultimately periodic
$\om$-strings $UP(Q)$ ($=UP_\cA(Q)$)
and   say that  a move $(q_0,t_0)(q_1.t_1)  \dots (q_i,t_i) \dots $  is simple if    $q_0q_1 \dots q_i \dots\in UP(Q)$ and $t_{i+1}=t_i+\Delta$ for  some $\Delta$. A strategy is simple if it uses
only simple moves.
We have
\begin{lem}[Reduction to simple strategies]
  \O has a winning strategy if and only if \O has a simple winning strategy.
\end{lem}

Now, we will play explicitly on a   finite (node) arena. This game corresponds to the game when \O is restricted to use only simple strategies.

Let $\cA:=(Q,\Sigma_{in}\times \Sigma_{out}, \delta, q_{\mathit{init}},pr)$ be a parity  automaton.
We construct the following arena. It has more states than an arena for right continuous game, because we need to deal also with signals which are discontinuous from the right.

The nodes  (vertex) of \O are of the form $(q,a)$ or $(q,\dag,a)$ where $q\in Q_\cA$ and $a\in \Sigma_{in}$.
At a moment $t$ a  play is in $(q,a)$  if the automaton is in state $q$ and \I chooses $a$ to hold at the moment $t$.
From this state \O should  choose an output  at $t$ and move to a node of the form $(q',\dag)$.  When a play is at 
$(q,\dag,a)$  at $t$ it indicates that $\cA$ was at state $q$ and \I chooses for $a\in \Sigma_{in}$ to holds for an  interval of  positive length after $t$.
From this node \O should respond with a signal on $(t,\infty)$. Our arena $\cG^{arena}_\cA$ is described below.

\fbox{
\begin{minipage}{0.9\linewidth}

\begin{center}
  Arena $\cG^{arena}_\cA$
\end{center}

\textbf{Nodes:} $\{q_{\new} \}\cup Q\times \Sigma_{in}\cup Q\times\{\dag\} \cup Q\times \{\dag\}\times\Sigma_{in} \cup  Q\times \Sigma_{in}\times UP(Q)$, where $Q_{\mathcal I}:=\{q_{\new}\}  \cup Q\times \{\dag\}\times\Sigma_{in} \cup  Q\times \Sigma_{in}\times UP(Q)$ is  the set of \I nodes,  and the rest is the set of \O nodes.

The set $F$  of final nodes  is $\{(q,\dag,u)\mid $ the maximal priority that appears infinitely often in $u$ is
even$\}$.

\textbf{Edges:}
\begin{enumerate}
  \item $q_{new} $ to $(q_{init},a)$ for every $a\in \Sigma_{in}$.

  \item $(q,a)$ to $(q',\dag) $ if $q'=\delta(q,(a,b))$ for some $b\in \Sigma_{out}$. 
  \item   $(q,\dag) $ to $(q ,\dag,a)$  for every $a\in \Sigma_{in}$ (here \I chooses $a$ to hold for a while).
   \item $(q_0,\dag, a)$ to  $(q_0,a,u)$, where   $ u=(q_1, q_2,\dots)\in UP(Q)$  such that $\exists b_i\in \Sigma_{out}. ~q_{i+1}=\delta(q_i,(a,b_i))$ for all $i$.
  \item There is an edge from $(q,a, u) $ to $(q',b)$ for $u\in UP(Q)$, 
   if $b\neq a$  and $q'=u(n)$ for $n=2m+1$ (this corresponds to discontinuity from the left).
   \item There is an edge from $(q,a, u) $ to $(q',\dag,b)$ for $u\in UP(Q)$, 
   if $b\neq a$  and $q'=u(n)$ for $n=2m$ (this corresponds to discontinuity from the  right).

\end{enumerate}
In (5) and (6) the edge will be labeled  by $p$ - the maximal priority among   $u(1),\dots , u(n)$  and it  will be called small if   $u(1),\dots , u(n)$  is in the lag part of $u$; otherwise, it will be called big.
The nodes also have priorities which are inherited from their $Q$ components, i.e.,
$pr(q,a)=pr(q,a,u)=pr(q,\dag)=pr(q,\dag,a)=pr_\cA(q)$.

\end{minipage}
}

\bigskip
Now, consider the following game.

\bigskip
\fbox{

\begin{minipage}{0.9\linewidth}

\begin{center}
  Timed game on  $\cG^{arena}_\cA$
\end{center}

Initial step. \I  chooses $a\in \Sigma_{in}$ and the play  proceeds to $(q_{init},a)$;  $t_0$ is set to 0.

When the play is at  node $(q,a)$,  \O chooses an edge to $(q',\dag)$ and the play is now at $(q',\dag)$.

From  $(q,\dag)$,  \I can move to  $(q,\dag,a)$.

From   $(q,\dag, a)$,  \O chooses an edge to $(q,a,u )$ and a time scale $\tau_i$. It defines a signal  $Y_i$ on $(t_i,\infty)$ as
$Y(t_i+\tau_i\times n)=u(2n)$ and $Y_i(t)=u(2n+1)$  if $t\in (t_i+\tau_i\times n, t_i+\tau_i\times (n+1))$.

From  a state $ (q,a, u)$  at $i+1$-move Player \I either (A) accepts and the play  ends, %It is winning for \O iff $ (q,a ,u)\in F$, 
or (B)
interrupts at $t_{i+1}>t_i$ by choosing  $b\neq a$ and the left or the right discontinuity at $t_{i+1}$. 
Let $q$ be  the left limit of $Y_i$ at $t_{i+1}$.

If the left discontinuity is chosen, then the play moves to $(q,b)$. 
If the right discontinuity is chosen, the play moves to $ (q, \dag,b)$. %A new round starts.

  \emph{Winning Condition.}
  \O wins a play $\pi$  if either $\pi$ is finite (this happens if \I has chosen option (A)) and its last state is in $F$ or $\pi$ is infinite and $\lim t_i$ is finite, or
  the maximal  priority that appears infinitely often %at the edges or the nodes of the play 
  is even.
  \end{minipage}
  }
  
   \begin{lem}\label{lem:fv-arena}
    \O wins $\cG^{arena}_\cA$ if and only if \O has a winning   simple strategy for $\cG_\cA$-game.
  \end{lem}

Next, like in the proof of Theorem \ref{th:dec-synth-rc}, we are going to restrict further the strategies of \O  in $ \cG^{arena}_\cA$.

 A strategy is called almost positional, if at every node  it always chooses the same untimed move (time scales might be different).

 \begin{lem}[Almost positional winning strategy in $\cG^{arena}_\cA$]

   \O has a simple winning strategy if and only if \O has an almost positional simple winning strategy.
 \end{lem}
 
Finally, we restrict  the time scale that  Player \O chooses.
We say that a strategy of \O is well-scaled if the  $i$-th time move from  a state of the form $(q,\dag,a)$ has the scale $2^{-i}$.
   \begin{lem}[$2^{-i}$ scale]
    \O has an almost positional simple winning strategy if and only if
    \O has a well-scaled almost positional simple winning strategy.

 \end{lem}

\begin{lem}[Deciding whether a well-scaled  strategy is winning] Let $\sigma$ be a well-scaled  almost positional winning strategy for \Op.
It is decidable whether $\sigma$ is winning.
\end{lem}
Since there are only finitely many well-scaled  almost positional winning strategies, we obtain that
 \begin{lem}\label{lem:fv-last} 
It is decidable whether \O has a winning strategy in  timed game on $\cG^{arena}_\cA$.
\end{lem}

Finally, Theorem \ref{mainthm} follows from Lemmas  \ref{lem:red-fv}-\ref{lem:fv-arena} and Lemma \ref{lem:fv-last}.

\section{Conclusion} \label{sect:conc}

We investigated  a generalization of the Church synthesis problem to  the continuous time domain  of the non-negative reals.
 We demonstrated that
in    continuous time there are phenomena  which are  very different from
the canonical discrete time domain of
 natural numbers. We proved  the decidability
the Church synthesis problem when the specifications are given by  $\MLO$ formulas. 

One can consider a specification language stronger than $\MLO$.  For example, let
 ${ \mathcal M}  : = \langle \reals,SIG,<_R, P_1,\dots ,P_k\rangle$ be an extension of $Sig$ structure by unary predicates $P_1,\dots, P_k$.
  If $MSO[<,P_1,\dots,P_k]$ decidable, then the   Church synthesis problem for specifications given
 by    $MSO[<,P_1,\dots,P_k]$  formulas is decidable. 

The extension of   $\MLO$  by $+1$ function is undecidable. Hence, the    Church synthesis problem for $MSO[<,+1]$  is undecidable.  There are decidable fragments of  $MSO[<,+1]$  which use metric constraints in a more restricted
form. The most widespread such  decidable metric language is Metric Interval Temporal Logic (MITL) \cite{AFH96}.
 We do not know whether the Church problem is decidable for MITL specifications.
 Proof techniques from \cite{DGRR09} might be usefull to show decidability/undecidability of the  Church synthesis problem for sub-logics of
 MITL.
 
 The satisfiability problem for $\MLO$ is non-elementary. The complexity of our   algorithm for the Church synthesis problem  is non-elementary because  
    any translation
of    $\MLO$-formulas  to  equivalent parity automata  is non-elementary.
  Another reason for the high complexity of the algorithm is that \O chooses her  moves from $UP(Q)$ and the size of $UP(Q)$ is triple exponential  (in the size of $Q$).
 
% 
%Fattal considered specifications (of signal languages)  given by deterministic parity automata \cite{Fat23}.
%He reduced the Church problem to  games in which \O chooses on her move  only one letter from $\Sigma_{out}$ and its duration.
%He provided an NP upper bound for the synthesis problem when the specifications  are provided by deterministic parity automata.
%For specifications given by deterministic Büchi automata  he provided a polynomial algorithm.
%Unfortunately, even for specifications given by deterministic automata of  parity 3 we do not know
%whether there exists  a polynomial algorithm (in the discrete case for every fixed bound $p$ on the priorities of a parity automaton, there s a polynomial algorithm).
Fattal considered specifications of signal languages given by deterministic parity automata \cite{Fat23}. He reduced the Church problem to games in
 which  
 \O
  selects, on her move, a single letter from the output alphabet 
   along with its duration. He provided an NP upper bound for the synthesis problem when the specifications are given by deterministic parity automata.
For   deterministic Büchi automata specifications, Fattal found a polynomial-time algorithm. However, even for specifications given  by deterministic parity automata with  three  priorities, it remains unknown whether a polynomial-time algorithm exists. (In the discrete case, 
for every fixed bound $p$ on the priorities of a parity automaton, there is a polynomial algorithm.)

\section*{Acknowledgement}
%I would like to thank anonymous referees for their insightful comments.

We  would like to express our  gratitude to Ori Lahav and the anonymous referees for their insightful comments.

\bibliographystyle{alphaurlsortkey}
\bibliography{MyBib1}
 
\appendix  % I moved the appendix to the end of the document. If you do not agree with this, let us know
\section{Proof of  Lemma \ref{lem:prop-equiv}(3)}
%\input{append}
%A coloring of a set $I$ is a function $f$ from the set of unordered pairs of $I$ into a finite set $T$ of colors
A coloring of  $\om$ is a function $f$ from the set of unordered pairs of distinct elements of  $\om$ into a finite set $T$ of colors.
For $x<y$, we will write $f(x,y)$ instead of $f(\{x,y\})$.
The  coloring $f$ is additive if $f(x,y)=f(x',y') $ and $ f(y,z)=f(y',z') $  imply
that $f(x,z)=f(x',z')$.  In this case a partial function $+$ is defined on $T$ such that 
if $x<y<z$, then $f(x,y)+f(y,z)=f(x,z)$.
A set $J$ is $t$-homogeneous (for $f$ and the  color $t$) %if there is a color $t$ such that $f(x,y)=t$ foe every $x<y\in J$.
if  $f(x,y)=t$ for every $x<y\in J$.

Ramsey's theorem for additive coloring \cite{She75} states:
\begin{thm}\label{th:ramsey} If $f$ is an additive coloring of $\om$, then there is $t\in T$ and  an infinite $t$-homogeneous subset $J$  of $\om$.
Moreover, there is $t_0\in T$ such that $f(0,j)=t_0$
for every $j\in J$.
%if $j$ is an  minimal element of $J$,  then $f(0,j)

\end{thm} 
Now, we are ready to prove  Lemma \ref{lem:prop-equiv}(3).

Let the set of colors   be the set of $\equiv_{state}$-equivalence classes. Given $\pi$.
Define coloring $f$ as $f(i,j):=$ the equivalence class of the substring of $\pi$ on the interval $[i,j)$.
By Lemma \ref{lem:prop-equiv}(2), this is an additive coloring.
Let $J=\{j_0<j_1<\cdots\} $ be an infinite homogeneous subset of color $v$ (such a set exists by Theorem \ref{th:ramsey}).

 Define $u_0$ to be the substring of $\pi$ over $[0,j_0)$, and $u_j$ to be the substring of $\pi$ over $[u_j,u_{j+1})$
 for $j>0$. Let $v_0$ be the equivalence class
of $u_0$.

 By Theorem \ref{th:ramsey} we obtain  $u_i$ and $u_iu_{i+1}$ and $u_j$  are all in the $v$-equivalence class for all $i,j>0$,
and $u_0\equiv_{state} u_0u_1$ are both in $v_0$ equivalence class.
Therefore,
  $u_i\equiv_{state} u_j \equiv_{state} u_iu_{i+1}$ for all $i,j>0$ and
   $u_0\equiv_{state} u_0u_1=u_0u_1\cdots u_j$. fot all $j\geq 0$.

\end{document}